\newif\ifreport\reporttrue
\documentclass[journal]{IEEEtran}

\usepackage{setspace}
\usepackage{cite}
\usepackage{graphicx}
\usepackage{tabularx}
\usepackage{epsfig}
\usepackage[T1]{fontenc}
\usepackage{epstopdf}
\usepackage{amsmath}
\usepackage{mathtools}
\usepackage{amsthm}
\usepackage{amsfonts}
\usepackage{amssymb}
\usepackage{algorithm}
\usepackage{algorithmic}
\usepackage{subfigure}
\usepackage{fixltx2e}
\usepackage{tikz,bm,hyperref}
\usetikzlibrary{arrows,snakes,shapes}
\usetikzlibrary{positioning,shapes}
\usetikzlibrary{decorations.pathreplacing}

\newcommand{\age}{\Delta}

\usepackage{color}

\newcommand{\ignore}[1]{}
\usepackage[singlelinecheck=false]{caption}

\newtheorem{lemma}{Lemma}

\newtheorem{theorem}{Theorem}
\newtheorem{corollary}{Corollary}

\begin{document}

\title{Sampling  for Data Freshness Optimization: Non-linear Age Functions}



%
\author{
Yin Sun and Benjamin Cyr
\thanks{This paper was presented in part at IEEE SPAWC 2018 \cite{SunSPAWC2018}.}
\thanks{This work was supported in part by NSF grant CCF-1813050 and ONR grant N00014-17-1-2417. Yin Sun  is  with  the  Department  of  ECE,  Auburn  University,  Auburn,  AL 36849 USA (e-mail:  yzs0078@auburn.edu). Benjamin Cyr is with the Department  of EECS,  University of Michigan, Ann Arbor, MI 48109 USA (e-mail:  bac0027@tigermail.auburn.edu).
Benjamin participated in this study when he was an undergraduate student at Auburn University.}
}
\maketitle

\begin{abstract}
In this paper, we study how to take sample at a data source for improving the freshness of received data samples at a remote receiver. 
We use non-linear functions of the age of information to measure data freshness, and provide a  survey of non-linear age functions and their applications. 
The sampler design problem for optimizing these data freshness metrics, possibly with a sampling rate constraint, is studied. This sampling problem is formulated as a constrained Markov decision process (MDP) with a possibly uncountable state space.
We present a {complete} characterization of the optimal  solution to this MDP: The optimal sampling policy is a deterministic or randomized threshold policy, where the threshold and the randomization probabilities are characterized  based on  the optimal objective value of the MDP and the sampling rate constraint. 
The optimal sampling policy can be computed by bisection search, and the curse of dimensionality is circumvented.
These age optimality results hold for (i) general data freshness metrics represented by monotonic functions of the age of information, (ii) general service time distributions of the queueing server, (iii) both continuous-time and discrete-time sampling problems, and (iv) sampling problems both with and without the sampling rate constraint.  Numerical results suggest that the optimal sampling policies can be much better than zero-wait sampling and the classic uniform sampling.

\end{abstract}

\section{Introduction}

Information usually has the greatest value when it is fresh \cite[p. 56]{Shapiro1999}. For example, real-time knowledge about the location, orientation, and speed of motor vehicles is imperative in autonomous driving, and the access to timely updates about the stock price and interest-rate movements is essential for developing trading strategies on the stock market. In \cite{Song1990,KaulYatesGruteser-Infocom2012}, the concept of \emph{Age of Information} was introduced  to measure the freshness of  information that a receiver has about the status of a remote source. Consider a sequence of source samples that are sent through a queue to a receiver. 
Let $U_t$ be the generation time of the newest sample that has been delivered to the receiver by time $t$. {The age of information, as a function of  $t$, is defined as} 
\begin{align}\label{eq_age_delta}
\Delta_t = t - U_t, 
\end{align}
which is the time elapsed since the newest sample was generated. 
Hence, a small  age $\Delta_t$ indicates that there exists a recently generated sample at the receiver. 



In practice, some information sources (e.g., vehicle location, stock price) vary quickly over time, while others (e.g., temperature, interest-rate) change slowly. Consider again the example of autonomous driving: The location information of motor vehicles collected 0.5 seconds ago could already be quite stale for making control decisions\footnote{A car will travel 15 meters during 0.5 seconds at the speed of 70 mph.}, but the engine temperature measured a few minutes ago is still valid for engine health monitoring. From this example, one can observe that data freshness should be evaluated based on \emph{(i)} the time-varying pattern of the source and \emph{(ii)} how valuable the fresh data is in the specific application. 
However, the age $\Delta_t$ defined in \eqref{eq_age_delta} is the time difference between data generation at the transmitter and data usage at the receiver, which cannot fully describe the source pattern and application context.\footnote{To the best of our knowledge, the issue that ``the actual age $\Delta_t$ is not a good representation of freshness'' was firstly pointed out by Anthony Ephremides in one presentation at the Information Theory and Application (ITA) Workshop in 2015.} This motivated us to seek more appropriate data freshness metrics that can interpret the role of freshness in real-time applications.



%
%

In this paper, we suggest to use a non-linear function $u(\Delta_t)$ of the age $\Delta_t$ as a data freshness metric, 
where $u(\Delta_t)$ could be the utility value of data with age $\Delta_t$, temporal autocorrelation function of the source, estimation error of signal value, or other application-specific performance metrics \cite{Cho:2003,Even:2007,Heinrich:2009,Ioannidis2009,Altman2011,Razniewski:2016,Soleymani2016-1,Soleymani2016-2,Soleymani2016-3,SunInfocom2016,AgeOfInfo2016,Kosta2017,SunSPAWC2018,Ornee2019, Champati2019,Markus2019}. 
A  survey of non-linear age functions and their applications is provided in Section \ref{sec_metrics}. Recently, the age of information has received significant  attention, because of the rapid deployment of real-time applications. A large portion of existing studies on  age  have been devoted to linear functions of the age $\Delta_t$, e.g., \cite{KaulYatesGruteser-Infocom2012,KamTIT2016,KamTIT2018,YatesTIT2018,
Bacinoglu2015,2015ISITYates,Bacinoglu2017,ArafaGLOBECOM2017,Wu2018,ArafaArXiv2018,FengArXiv2018,BacinogluISIT2018,Zhong2016,YatesISIT2017,Ceran2018,Mayekar2018,
HeTIT2018,KadotaAllerton2016,JooTON2018,HsuISIT2017,KadotaINFOCOM2018,Lu:2018,JiangISIT2018,Zhou2018,Xiao2018,Gopal2018,SoleymaniArXiv2018,Sonmez2018}. 
However, the design of efficient data update policies for optimizing non-linear age metrics remains largely unexplored.
To that end, we investigate a problem of sampling an information source, where the samples are forwarded to a remote receiver through a channel that is modeled as a FIFO queue. The optimal sampler design for optimizing non-linear age metrics is obtained. The contributions of this paper are summarized as follows:

\begin{itemize}

\item We consider a class of data freshness metrics, where the utility for data freshness is represented  by a  \emph{non-increasing} function $u(\Delta_t)$ of the age $\Delta_t$. Accordingly, the penalty for data staleness is denoted by a  \emph{non-decreasing}  function $p(\Delta_t)$ of $\Delta_t$. The sampler design problem for optimizing these data freshness metrics, possibly with a sampling rate constraint, is considered. This sampling problem is formulated as a constrained Markov decision process (MDP) with a possibly uncountable state space.

\item We prove that an optimal sampling solution to this MDP is a deterministic or randomized threshold policy, where 
the threshold is equal to the optimum objective value of the MDP plus the optimal Lagrangian dual variable associated with the sampling rate constraint; see Section \ref{sec_proof_2} for the details. 
The  threshold can be computed by bisection search, and
the randomization probabilities are chosen to satisfy  the sampling rate constraint.
The curse of dimensionality is circumvented in this sampling solution by exploiting the structure of the MDP. 
These age optimality results hold for \emph{(i)} general monotonic age metrics, \emph{(ii)} general service time distributions of the queueing server,  \emph{(iii)} both continuous-time and discrete-time sampling problems, and \emph{(iv)} sampling problems both with and without the sampling rate constraint. Among the technical tools used to prove these results are an extension of Dinkelbach's method for MDP and a geometric multiplier technique for establishing strong duality. These technical tools were recently used in \cite{SunISIT2017,SunTIT2018}, where a quite different sampling problem was solved. In addition, we will also introduce some proof ideas that are specific to the sampling problem that we consider in this paper, which will be 
used to prove Lemma \ref{lem_optimal_eq_opt_stopping22}, Theorem \ref{thm6_strong_duality}, and Lemma \ref{lem_optimal_eq_opt_stopping23} in Section \ref{sec_analysis}. 

\item When there is no  sampling rate constraint, a logical sampling policy is the zero-wait sampling policy \cite{KaulYatesGruteser-Infocom2012,2015ISITYates,AgeOfInfo2016}, which is  throughput-optimal and delay-optimal, but  not necessarily age-optimal. We develop sufficient and necessary conditions for characterizing  the optimality of the zero-wait sampling policy for general monotonic age metrics. Our numerical results show that the optimal sampling policies can be much better than zero-wait sampling and the classic uniform sampling. 

\end{itemize}

The rest of this paper is organized as follows. In Section \ref{sec_related},
we discuss some related work. In Section \ref{sec_model}, we describe the
system model and the formulation of the optimal sampling problem; a short survey of non-linear age functions is also provided. In Section \ref{sec_main_results}, we present the optimal sampling policy for different system settings, as well as a sufficient and necessary condition for the optimality of the zero-wait sampling policy. The proofs are provided In Section \ref{sec_analysis}. The numerical results and the conclusion are presented in Section \ref{sec:numerical} and Section \ref{sec_conclusion}.

\section{Related Work}\label{sec_related}
The age of information was used as a data freshness metric as early as 1990s in the studies of real-time databases \cite{Song1990,Segev:1991,Adelberg:1995,Cho:2000}. Queueing theoretic techniques were introduced to evaluate the age of information in \cite{KaulYatesGruteser-Infocom2012}. The average age, average peak age, and age distribution have been analyzed for various queueing systems in, e.g.,\cite{KaulYatesGruteser-Infocom2012,KamTIT2016,CostaTIT2016,KamTIT2018,YatesTIT2018,2015ISITHuangModiano,InoueISIT2017,Kosta2017,YatesArXiv2018}. 
It was observed that a Last-Come, First-Served (LCFS) scheduling policy can achieve a smaller time-average age  than a few other scheduling policies. The optimality of the LCFS policy, or more generally the Last-Generated, First-Served (LGFS) policy, was first proven in  \cite{Bedewy2016}. This age optimality result holds for several queueing systems with multiple servers, multiple hops, and/or multiple sources  \cite{Bedewy2016,Bedewy2017,BedewyJournal2017,BedewyJournal2017_2,multiflow18}. 

When the transmission power of the source is subject to an energy-harvesting constraint, the age of information was minimized in, e.g., \cite{Bacinoglu2015,2015ISITYates,AgeOfInfo2016,Bacinoglu2017,ArafaGLOBECOM2017,Wu2018,ArafaArXiv2018,FengArXiv2018,BacinogluISIT2018}. Source coding and channel coding schemes for reducing the age were developed in, e.g., \cite{Zhong2016,YatesISIT2017,Ceran2018,Mayekar2018}. Age-optimal transmission scheduling of wireless networks have been investigated in, e.g., \cite{HeTIT2018,KadotaAllerton2016,JooTON2018,HsuISIT2017,TalakAllerton2017,KadotaINFOCOM2018,Talak:2018,Lu:2018,JiangISIT2018,Zhou2018}. Game theoretical perspective of the age was studied in \cite{NguyenWiOpt2017,NguyenINFOCOM2018,Xiao2018,Gopal2018}. The aging effect of channel state information was analyzed in, e.g., \cite{Truong2013,CostaISIT2015,FaraziICASSP2016}. An interesting connection between the age of information and remote estimation error was revealed in \cite{SunISIT2017,SunTIT2018,Ornee2019}, where the optimal sampling policies were obtained for two continuous-time processes. 
The impact of the age to control systems was studied in \cite{ZhangISIT2018,SoleymaniArXiv2018,Champati2019,Markus2019}.
Emulations and measurements of the age were conducted in \cite{Kam2015,KamMilCom2017,Sonmez2018}. An age-based transport protocol was developed in \cite{Shreedhar:2018}. 

In \cite{Ceran2018,Zhou2018}, optimal sampling policies were developed to minimize the time-average age for status updates over wireless channels, where the optimal sampling policies were shown to be randomized threshold policies. Structural properties of the randomized threshold policies were obtained in \cite{Ceran2018,Zhou2018} to simplify the value iteration or policy iteration algorithms therein. The linear age function considered in \cite{Ceran2018,Zhou2018} is a special case of the monotonic age functions considered in this paper, and the channel models in \cite{Ceran2018,Zhou2018} are different from ours. In our study, the optimal sampling policies are characterized semi-analytically and can be computed by bisection search. 
 In a special case of \cite{Ceran2018}, a closed-form optimal sampling solution was obtained. However, it is unclear whether (semi-)analytical or closed-form solutions can be found for the general cases considered  in \cite{Ceran2018,Zhou2018}.

The most relevant prior study to this paper is \cite{AgeOfInfo2016}.  
This paper generalizes \cite{AgeOfInfo2016} in the following aspects: \emph{(i)} The data freshness metrics considered in this paper are more general than those  of \cite{AgeOfInfo2016}. The age penalty function $p(\age_t)$ in \cite{AgeOfInfo2016} is assumed to be non-negative and non-decreasing, which is relaxed in this paper to be an arbitrary non-decreasing function that is more desirable for some applications.
\emph{(ii)} The optimal sampling policies developed in this paper are  simpler and more insightful than those in \cite{AgeOfInfo2016}. A two-layered nested bisection search algorithm
was developed to compute the optimal threshold \cite{AgeOfInfo2016}. In this paper, the optimal threshold can be computed by a single layer of bisection search. \emph{(iii)} In \cite{AgeOfInfo2016}, the optimal sampling strategy was obtained for
continuous-time systems. In this paper, we also develop an optimal sampling strategy for discrete-time systems, without sacrificing from any approximation or sub-optimality. 
\emph{(iv)} It was assumed in \cite{AgeOfInfo2016} that after the previous sample was delivered, the next sample must be generated within a given amount of time. By adopting more insightful proof techniques, we are able to remove such an assumption and greatly simplify the proofs in this paper.

\section{ Model, Metrics, and Formulation} \label{sec_model}
\begin{figure}
\centering
\includegraphics[width=0.35\textwidth]{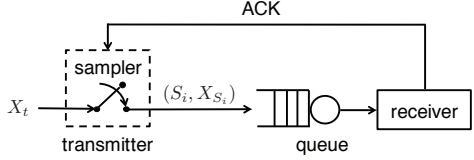}   
\caption{System model.}
\label{fig_model}
\end{figure}    
\subsection{System Model}

We consider the status update system illustrated in Fig. \ref{fig_model}, where samples of a source process $X_t$ are taken and sent to a receiver through a communication channel. The channel is modeled as a single-server FIFO queue with  \emph{i.i.d.} service times. The system starts to operate at time  $t=0$. The $i$-th sample is generated at time  $S_i$ and is delivered to the receiver at time  $D_i$ with a service time $Y_i$, which satisfy  $S_i\leq S_{i+1}$, $S_i +Y_i \leq D_i$, $D_i+Y_{i+1}\leq D_{i+1}$, and $0<\mathbb{E}[Y_i]<\infty$  for all $i$. Each sample packet $(S_i, X_{S_i})$ contains  the sampling time $S_i$ and the sample value $X_{S_i}$. 
Once a sample is delivered, the receiver sends  an acknowledgement (ACK) back to the sampler with zero delay. Hence, the sampler has access to the idle/busy state of the server  in real-time. 

Let $U_t= \max\{S_{i}: D_{i} \leq t\}$ be the generation time   of the freshest sample that has been delivered to the receiver by time  $t$. Then, the \emph{age of information}, or simply  \emph{age}, at time  $t$ is defined by  \cite{Song1990,KaulYatesGruteser-Infocom2012}
\begin{align}\label{eq_age}
\Delta_{t} = t-U_t = t- \max\{S_{i}: D_{i} \leq t\},
\end{align}
which is plotted in Fig. \ref{fig:age1}. 
Because $D_i \leq D_{i+1}$, $\age_t$ can be also written as
\begin{align}\label{eq_age2}
\age_t = t- S_i, \text{ if } D_i \leq t < D_{i+1}.
\end{align}
The initial state of the system is assumed to be $S_0 = 0$, $D_0 = Y_0$, and $\Delta_{0}$ is a finite constant. 

In this paper, we will consider both continuous-time and discrete-time status-update systems. In the continuous-time setting, $t\in[0,\infty)$ can take any positive value. In the discrete-time setting, $t\in\{0,T_s, 2T_s,\ldots\}$ is a multiple of period $T_s$; as a result, $S_i, D_i, Y_i, t,U_t, \Delta_{t}$ are all discrete-time variables. For notational simplicity, we choose $T_s=1$ second such that all the discrete-time variables are integers. The results for other values of $T_s$ can be readily obtained  by time scaling. 

In practice, the continuous-time setting can be used to model status-update systems with a high clock rate, while the discrete-time setting is appropriate for characterizing  sensors that have a very low energy budget and can only wake up periodically from a low-power sleep mode.


\subsection{Data Staleness and Freshness Metrics: A  Survey}\label{sec_metrics}

The dissatisfaction for data staleness (or the eagerness for data refreshing) is represented by 
a penalty function $p(\Delta)$ of the age $\Delta$, where the function $p: [0,\infty) \mapsto \mathbb{R}$ is \emph{non-decreasing}. This non-decreasing requirement on $p(\Delta)$ complies with the observations that stale data is usually less desired than fresh data \cite{Shapiro1999,Cho:2003,Even:2007,Heinrich:2009,Ioannidis2009,Altman2011,Razniewski:2016}. 
This data staleness model is quite general, as it allows $p(\Delta)$ to be non-convex or  dis-continuous. These data staleness metrics are clearly more general than those in \cite{SunInfocom2016,AgeOfInfo2016}, where $p(\Delta)$ was restricted to be \emph{non-negative} and \emph{non-decreasing}.
 
Similarly, data freshness can be characterized by a \emph{non-increasing} utility function $u(\Delta)$ of the age $\Delta$ \cite{Even:2007,Ioannidis2009}. One simple choice is $u(\Delta)=-p(\Delta)$. 
Note that because the age $\Delta_t$ is a function of time $t$,  $p(\Delta_t)$ and $u(\Delta_t)$ are both time-varying, as illustrated in Fig. \ref{fig:age2}. 
In practice, one can choose $p(\cdot)$ and $u(\cdot)$ based on the information source and the application under consideration, as illustrated in the following examples.\footnote{In some  of these examples, the age utility function $u(\age_t)$ is non-negative and non-increasing. The corresponding age penalty function $p(\age_t) = - u(\age_t)$ is non-positive and non-decreasing. Hence, it is desirable to allow the age penalty function $p(\age_t)$ to be negative.}  

\subsubsection{Auto-correlation Function of the Source}
The auto-correlation function $\mathbb{E}[X_t^* X_{t-\Delta_t}]$ can be used to evaluate the freshness of the sample $X_{t-\Delta_t}$ \cite{Kosta2017}. For some stationary sources, $|\mathbb{E}[X_t^* X_{t-\Delta_t}]|$ is a non-negative, non-increasing function of the age $\Delta_t$, which can be considered as an age utility function $u(\Delta_t)$.
For example, in stationary ergodic Gauss-Markov block fading channels, the impact of channel aging can be characterized by the auto-correlation function of fading channel coefficients. When the age $\Delta_t$ is  small, the auto-correlation function and the data rate both decay with respect to the age $\Delta_t$ \cite{Truong2013}.
\begin{figure}
\centering
\begin{tikzpicture}[scale=0.21]
\draw [<-|] (0,11)  -- (0,0) -- (14.5,0);
\draw [|->] (15,0) -- (30.5,0) node [below] {\small$t$};
\draw (-2,12) node [right] {\small$\age_t$};
\draw
(0,0) node [below] {\small$S_0$}
(8,0) node [below] {\small$S_1$}
(17,0) node [below] {\small$S_{j-1}$}
(24,0) node [below] {\small$S_{j}$};
\fill
(8,0)  circle[radius=4pt]
(17,0)  circle[radius=4pt]
(24,0)  circle[radius=4pt]
(0,0)  circle[radius=4pt]
(4,0)  circle[radius=4pt]
(11,0)  circle[radius=4pt]
(20,0)  circle[radius=4pt]
(27,0)  circle[radius=4pt];
\draw
(4,0) node [below] {\small$D_0$}
(11,0) node [below] {\small$D_1$}
(21,0) node [below] {\small$D_{j-1}$}
(27,0) node [below] {\small$D_{j}$};
\draw[ thick, domain=0:4] plot (\x, {\x+3})  -- (4, {4});
 \draw [ thick, domain=4:11] plot (\x, {\x})  -- (11, {3});
\draw[ thick, domain=11:14] plot (\x, {\x-8});
\draw[ thick, domain=16:20] plot (\x, {\x-14}) -- (20, {3});
\draw[ thick, domain=20:27] plot (\x, {\x-17}) -- (27, {3});
\draw[ thick, domain=27:29] plot (\x, {\x-24});
\draw[  thin,dashed,  domain=0:4] plot (\x, {\x})-- (4, 0);
\draw[  thin,dashed,  domain=8:11] plot (\x, {\x-8})-- (11, 0);
\draw[  thin,dashed,  domain=17:20] plot (\x, {\x-17})-- (20, 0);
\draw[  thin,dashed,  domain=24:27] plot (\x, {\x-24})-- (27, 0);
\end{tikzpicture}
\caption{Evolution of the age $\age_t$ over time.}
\label{fig:age1}
\vspace{-0.5cm}
\end{figure}
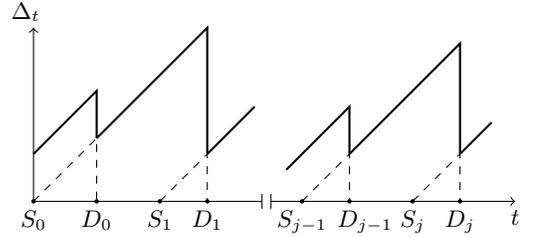

\subsubsection{Estimation Error of Real-time Source Value} 
Consider a status-update system, where samples of a Markov source $X_t$ are forwarded to a remote estimator. The estimator uses causally received samples to reconstruct an estimate $\hat X_t$ of  real-time source value. If the sampling times $S_i$ are independent of the observed source $\{X_t, t\geq 0\}$, 
the mean-squared estimation error at time $t$ can be expressed as an age penalty function $p(\age_t)$ \cite{YatesTIT2018,SunISIT2017,SunTIT2018,Ornee2019}. 
If the sampling times $S_i$ are chosen based on causal knowledge about the source, the estimation error is not a  function of $\Delta_t$ \cite{SunISIT2017,SunTIT2018,Ornee2019}.

The above result can be generalized to the state estimation error of feedback control systems \cite{Champati2019,Markus2019}. Consider a single-loop feedback control system, where a plant and a controller are governed by a Linear Time-Invariant (LTI) system, i.e., 
\begin{align}
X_{t+1} = A X_t + B U_t + N_t,
\end{align}
where $X_t \in \mathbb{R}^n$ is the state of the system at time slot $t$, $n$ is the system dimension, $U_t\in \mathbb{R}^m$  represents the control input, and $N_t \in \mathbb{R}^n$ is the exogenous noise vector having \emph{i.i.d.} Gaussian distributed elements with zero mean  and covariance $\Sigma$. 
The constant matrices $ A\in \mathbb{R}^{n\times n}$  and $B\in \mathbb{R}^{n\times m}$ are the system and input matrices, respectively, where $(A,B)$ is assumed to be controllable. 
Samples of the state process $X_t $ are forwarded to the controller, which determines $U_t$ at time $t$ based on the samples that have been delivered by time $t$. Under some assumptions, the state estimation error can be proven to be independent of the adopted control policy \cite{SoleymaniArXiv2018}. 
Furthermore, if the sampling times $S_i$ are independent of the state process $X_t$, then the state estimation error is an age penalty function $p(\Delta_t)$ that is determined by the system matrix $A$ and the covariance $\Sigma$ of the exogenous noise \cite{Champati2019,Markus2019}.


\begin{figure}
\centering
\subfigure[][Non-decreasing age penalty function $p(\age_t)=e^{0.2\age_t}-1$.]
{
\begin{tikzpicture}[scale=0.21]
\draw [<-|] (0,8)  -- (0,0) -- (14.5,0);
\draw [|->] (15,0) -- (30.5,0) node [below] {\small$t$};
\draw (-3,9) node [right] {\small$p(\age_t)=e^{0.2\age_t}-1$};
\fill
(8,0)  circle[radius=4pt]
(17,0)  circle[radius=4pt]
(24,0)  circle[radius=4pt];
\draw
(0,0) node [below] {\small$S_0$}
(8,0) node [below] {\small$S_1$}
(17,0) node [below] {\small$S_{j-1}$}
(24,0) node [below] {\small$S_{j}$};
\fill
(0,0)  circle[radius=4pt]
(4,0)  circle[radius=4pt]
(11,0)  circle[radius=4pt]
(20,0)  circle[radius=4pt]
(27,0)  circle[radius=4pt];
\draw
(4,0) node [below] {\small$D_0$}
(11,0) node [below] {\small$D_1$}
(21,0) node [below] {\small$D_{j-1}$}
(27,0) node [below] {\small$D_{j}$};
\draw[ thick, domain=0:4] plot (\x, {exp(0.2*(\x+3))-1}) -- (4, {exp(0.2*(4))-1});
\draw[ thick, domain=4:11] plot (\x, {exp(0.2*(\x))-1}) -- (11, {exp(0.2*(11-8))-1});
\draw[ thin,dashed,  domain=0:4] plot (\x, {exp(0.2*(\x))-1}) -- (4,0);
\draw[ thick, domain=11:14] plot (\x, {exp(0.2*(\x-8))-1});
\draw[ thin,dashed,  domain=8:11] plot (\x, {exp(0.2*(\x-8))-1}) -- (11,0);
\draw[ thick, domain=16:20] plot (\x, {exp(0.2*(\x-14))-1}) -- (20, {exp(0.2*(20-17))-1});
\draw[ thick, domain=20:27] plot (\x, {exp(0.2*(\x-17))-1}) -- (27, {exp(0.2*(27-24))-1});
\draw[ thin,dashed,  domain=17:20] plot (\x, {exp(0.2*(\x-17))-1}) -- (20,0);
\draw[ thick, domain=27:29] plot (\x, {exp(0.2*(\x-24))-1});
\draw[ thin,dashed,  domain=24:27] plot (\x, {exp(0.2*(\x-24))-1}) -- (27,0);
\end{tikzpicture}
}

\subfigure[][Non-increasing age utility function $u(\age_t)=10/\age_t$.]
{
\begin{tikzpicture}[scale=0.21]
\draw [<-|] (0,7)  -- (0,0) -- (14.5,0);
\draw [|->] (15,0) -- (30.5,0) node [below] {\small$t$};
\draw (-3,9) node [right] {\small$u(\age_t)=10/\age_t$};
\fill
(8,0)  circle[radius=4pt]
(17,0)  circle[radius=4pt]
(24,0)  circle[radius=4pt];
\draw
(0,0) node [below] {\small$S_0$}
(8,0) node [below] {\small$S_1$}
(17,0) node [below] {\small$S_{j-1}$}
(24,0) node [below] {\small$S_{j}$};
\fill
(0,0)  circle[radius=4pt]
(4,0)  circle[radius=4pt]
(11,0)  circle[radius=4pt]
(20,0)  circle[radius=4pt]
(27,0)  circle[radius=4pt];
\draw
(4,0) node [below] {\small$D_0$}
(11,0) node [below] {\small$D_1$}
(21,0) node [below] {\small$D_{j-1}$}
(27,0) node [below] {\small$D_{j}$};
\draw[ thick, domain=0:4] plot (\x, {10/(\x+3))}) -- (4, {10/4});
\draw[ thin, dashed] (4, {10/7}) -- (4,0);
\draw[ thin, dashed, domain=1.5:4] plot (\x, {10/(\x)});
\draw[ thick, domain=4:11] plot (\x, {10/(\x)}) -- (11, {10/(11-8)});
\draw[ thin, dashed] (11, {10/10}) -- (11,0);
\draw[ thin, dashed, domain=9.5:11] plot (\x, {10/(\x-8)});
\draw[ thin, dashed] (8,0)--(8,7);
\draw[ thin, dashed] (17,0)--(17,7);
\draw[ thin, dashed] (24,0)--(24,7);
\draw[ thick, domain=11:14] plot (\x, {10/(\x-8)});
\draw[ thick, domain=16:20] plot (\x, {10/(\x-14)}) -- (20, {10/(20-17)});
\draw[ thin, dashed] (20, {10/6}) -- (20,0);
\draw[ thin, dashed, domain=18.5:20] plot (\x, {10/(\x-17)});
\draw[ thick, domain=20:27] plot (\x, {10/(\x-17)}) -- (27, {10/(27-24)});
\draw[ thin, dashed, domain=25.5:27] plot (\x, {10/(\x-24)});
\draw[ thick, domain=27:29] plot (\x, {10/(\x-24)});
\draw[ thin, dashed] (27, {10/10}) -- (27,0);
\end{tikzpicture}
}

\caption{Two examples of non-linear age functions.}
\vspace{-0.5cm}
\label{fig:age2}
\end{figure}
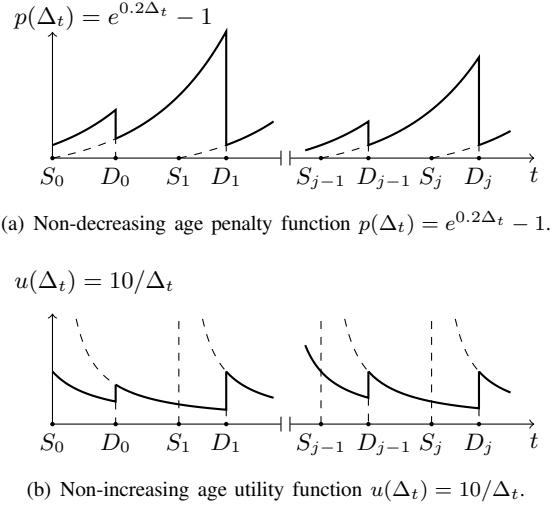

\subsubsection{Information based  Data Freshness Metric} 
Let
\begin{align}\label{eq_samples}
\bm{W}_t = \{(X_{S_i},S_i): D_i \leq t\}
\end{align}
denote the samples that have been delivered to the receiver  by time $t$. One can use the mutual information $I(X_t; \bm{W}_t)$ --- the amount of information that the received samples $\bm{W}_t$ carry about the current source value $X_t$ --- to evaluate the freshness of $\bm{W}_t$. 
If $I(X_t; \bm{W}_t)$ is close to $H(X_t)$, the samples $\bm{W}_t$ contains a lot of information about $X_t$ and is considered to be fresh; if $I(X_t; \bm{W}_t)$ is almost $0$, $\bm{W}_t$ provides little information about $X_t$ and is deemed to be obsolete. 

One way to interpret  $I(X_t; \bm{W}_t)$ is to consider how helpful the received samples $\bm{W}_t$ are for inferring $X_t$. 
By using the Shannon code lengths \cite[Section 5.4]{Cover}, the expected minimum number of bits $L$ required to specify $X_t$ satisfies 
\begin{align}\label{eq_length}
H(X_t) \leq L < H(X_t) + 1,
\end{align}
where $L$ can be interpreted as the expected minimum number of binary tests that are needed to infer $X_t$. 
On the other hand, with the knowledge of $\bm{W}_t$, the expected minimum number of bits $L'$ that are required to specify $X_t$ satisfies
\begin{align}\label{eq_length1}
H(X_t| \bm{W}_t) \leq L' < H(X_t| \bm{W}_t) + 1.
\end{align}
If $X_t$ is a random vector consisting of a large number of symbols (e.g., $X_t$ represents an image containing many pixels or the coefficients of MIMO-OFDM channels), the one bit of overhead in \eqref{eq_length} and \eqref{eq_length1} is insignificant. 
Hence, $I(X_t; \bm{W}_t)$ is approximately the reduction in the description cost for inferring $X_t$ without and with the knowledge of $\bm{W}_t$. 

If $X_t$ is a stationary Markov chain, by data processing inequality \cite[Theorem 2.8.1]{Cover},
it is easy to prove the following lemma:

\begin{lemma}\label{lem1}
If $X_t$ is a stationary (continuous-time or discrete-time) Markov chain, $\bm{W}_t$ is defined in \eqref{eq_samples}, and the sampling times $S_i$ are independent of $\{X_t, t\geq0\}$, then the mutual information
\begin{align}\label{eq_lem1}
I(X_t; \bm{W}_t) = I(X_t; X_{t-\age_t})
\end{align}
is a non-negative and non-increasing function $u(\age_t)$ of $\age_t$.
\end{lemma}
\begin{proof}
See Appendix \ref{app_lem1}. 
\end{proof}

Lemma \ref{lem1} provides an intuitive interpretation of ``information aging'': The amount of information $I(X_t; \bm{W}_t)$ that is preserved in $\bm{W}_t$ for {inferring} the current source value $X_t$ decreases as the age $\age_t$ grows. We note that Lemma \ref{lem1} can be generalized to the case that $X_t$ is a stationary discrete-time Markov chain with memory $k$. In this case, each sample $\bm V_{t} = (X_{t}, X_{t-1}, \ldots, X_{t-k+1})$ should contain the source values at $k$ successive time instants. Let
 $\bm W_t = \{(\bm V_{S_i},S_i): D_i \leq t\}$, then one can show that $\bm V_{t-\Delta_t}$ is a sufficient statistic of $\bm{W}_t$ for inferring $X_t$ and $I(X_t; \bm{W}_t) = I(X_t; \bm V_{t-\age_t})$ is  a non-negative and non-increasing function of $\age_t$. 

If the sampling times $S_i$ are determined by using causal knowledge of $X_t$, $I(X_t; \bm{W}_t)$ is not necessarily a function of the age. One interesting future research direction is how to choose the sampling time $S_i$ based on the signal and utilize the timing information in $S_i$ to improve data freshness.

Next, we provide the closed-form expression of $I(X_t; \bm{W}_t)$ for 
two Markov sources:


\emph{Gauss-Markov Source:} 
Suppose that $X_t$ is a first-order discrete-time Gauss-Markov process, defined by 
\begin{align} \label{eq_markov}
X_t = a X_{t-1} + V_t,
\end{align} 
where $a\in(-1,1)$ and the $V_t$'s are zero-mean \emph{i.i.d.}~Gaussian random variables with variance $\sigma^2$. Because $X_t$ is a Gauss-Markov process, one can show that \cite{Gelfand1959}
\begin{align}\label{eq_MI_Gauss_Markov}
I(X_t; \bm{W}_t) = I\left(X_t; X_{t-\age_t}\right) = -\frac{1}{2} \log_2 \left(1-a^{2\age_t}\right).
\end{align} 
Since $a\in(-1,1)$ and $\age_t\geq 0$ is an integer, $I(X_t; \bm{W}_t)$ is a positive and decreasing function of the age $\age_t$. Note that if $\age_t=0$, then $I(X_t; \bm{W}_t)=H(X_t)=\infty$, because the absolute entropy of a Gaussian random variable is infinite. 

\emph{ Binary Markov Source:} Suppose that $X_t\in\{0,1\}$ is a binary symmetric Markov process defined by
\begin{align}\label{eq_binary}
X_t = X_{t-1} \oplus V_t,
\end{align} 
where $\oplus$ denotes binary modulo-2 addition and the $V_t$'s are \emph{i.i.d.} Bernoulli random variables with mean $q\in[0,\frac{1}{2}]$. One can show that
\begin{align}
I(X_t; \bm{W}_t) &= I(X_t; X_{t-\age_t}) = 1\!-\!h\!\left(\frac{1\!-\!(1-2q)^{\age_t} }{2}\right)\!,\!\!
\end{align} 
where $\Pr[X_t =1 | X_0=0] = \frac{1-(1-2q)^{t} }{2}$ and $h(x)$ is the binary entropy function defined by $h(x)=-x\log_2 x - (1-x) \log_2 (1-x)$
with a domain $x\in[0,1]$ \cite[Eq. (2.5)]{Cover}. Because $h(x)$ is increasing on $[0,\frac{1}{2}]$, $I(X_t; \bm{W}_t)$ is a non-negative and decreasing function of the age $\age_t$.


Similarly, one can also use the conditional entropy $H(X_t| \bm{W}_t)$ to represent the staleness of $\bm{W}_t$ \cite{Soleymani2016-1,Soleymani2016-2,Soleymani2016-3}. In particular, 
$H(X_t| \bm{W}_t)$ can be interpreted as the amount of uncertainty about the  current source value $X_t$ after receiving the samples $\bm{W}_t$. 
 If the $S_i$'s are independent of $\{X_t, t \geq 0\}$ and $X_t$ is a stationary Markov chain, $H(X_t| \bm{W}_t)= H(X_t| \{X_{S_i}:D_i\leq t\})=H(X_t| X_{{t-\age_t}})$  is a non-decreasing function $p(\age_t)$ of the age $\age_t$. 
If the sampling times $S_i$ are determined based on causal knowledge of $X_t$ or $X_t$ is not a Markov chain, $H(X_t| \bm{W}_t)$ is no longer a function of the age.

More usage cases of $p(\cdot)$ and $u(\cdot)$ can be found in \cite{Cho:2003,Even:2007,Heinrich:2009,Ioannidis2009,Altman2011,Razniewski:2016}. Other data freshness metrics that cannot be expressed as functions of    $\age_t$ were discussed in \cite{CostaTIT2016,Bedewy2016,Bedewy2017,BedewyJournal2017,BedewyJournal2017_2,multiflow18}.

\subsection{Formulation of Optimal Sampling Problems}

Let $\pi = (S_1,S_2,\ldots)$ represent a sampling policy 
 and $\Pi$ denote the set of  \emph{causal} sampling policies that satisfy the following two conditions: \emph{(i)} Each sampling time $S_i$ is chosen based on {history and current information of the idle/busy state of the channel}. \emph{(ii)} The inter-sampling times $\{T_i = S_{i+1}-S_i, i=1,2,\ldots\}$ form a {regenerative process} \cite[Section 6.1]{Haas2002}\footnote{We assume that $T_i$ is a regenerative process because we will optimize $\limsup_{T\rightarrow \infty}\mathbb{E}[\int_{0}^T p(\age_t) dt]/T$, but operationally a nicer objective function is $\limsup_{i\rightarrow \infty}\mathbb{E}[\int_{0}^{D_i} p(\age_t) dt]/{\mathbb{E}[D_i]}$. These two criteria are equivalent, if $\{T_1,T_2,\ldots\}$ is a regenerative process, or more generally, if $\{T_1,T_2,\ldots\}$ has only one ergodic class. If no condition is imposed, however, they are different.}: There exists an increasing sequence  $0\leq {k_1}<k_2< \ldots$ of almost surely finite random integers such that the post-${k_j}$ process $\{T_{k_j+i}, i=1,2,\ldots\}$ has the same distribution as the post-${k_1}$ process $\{T_{k_1+i}, i=1,2,\ldots\}$ and is independent of the pre-$k_j$ process $\{T_{i}, i=1,2,\ldots, k_j-1\}$; in addition, $\mathbb{E}[{k_{j+1}}-{k_j}]<\infty$, $\mathbb{E}[S_{k_{1}}]<\infty$, and $0<\mathbb{E}[S_{k_{j+1}}-S_{k_j}]<\infty, ~j=1,2,\ldots$ 
 
 We assume that the sampling times $S_i$ are independent of  the source process $\{X_t, t\geq 0\}$, and
  the service times $Y_i$ of the queue do not change according to the  sampling policy.  We further assume that $\mathbb{E} [p(\Delta+Y_i)]<\infty$ for all  finite $\Delta$.



In this paper, we study the optimal  sampling policy that minimizes (maximizes) the average age penalty (utility) subject to an average sampling rate constraint. In the continuous-time case, we will consider the following problem: 
\begin{align}
\bar p_{\text{opt,1}} = \inf_{\pi\in\Pi}~&\limsup_{T\rightarrow \infty}\frac{1}{T}~\mathbb{E}\left[\int_{0}^T p(\age_t) dt\right]\label{eq_problem2} \\
\text{s.t.}~~& \liminf_{n\rightarrow \infty} \frac{1}{n}{\mathbb{E}\left[S_{n}\right]} \geq \frac{1}{f_{\max}},\label{eq_constraint2}
\end{align}
where $\bar p_{\text{opt,1}}$ is the optimal value of  \eqref{eq_problem2} and $f_{\max}$ is the maximum allowed sampling rate.
\begin{algorithm}
\caption{Bisection method for solving \eqref{thm1_eq22}} \label{alg1}
\begin{algorithmic}[]
\STATE \textbf{given} $l$, $u$, tolerance $\epsilon>0$.
\REPEAT
\STATE $\beta:= (l+u)/2$.
\STATE $o :=\beta-\frac{\mathbb{E}\left[v(D_{i+1}(\beta)-S_i(\beta))-v(Y_i)\right]}{\mathbb{E}[D_{i+1}(\beta)- D_i(\beta)]}$.
\STATE \textbf{if} $o\geq 0$, $u:=\beta$; \textbf{else}, $l:=\beta$.
\UNTIL $u-l\leq \epsilon$.
\STATE \textbf{return} $\beta$.
\end{algorithmic}
\end{algorithm}
In the discrete-time case, we need to solve the following optimal sampling problem:
\begin{align}
\bar p_{\text{opt,2}} = \inf_{\pi\in\Pi}~&\limsup_{n\rightarrow \infty}\frac{1}{n}~\mathbb{E}\left[\sum_{t=1}^n p(\age_{t})\right]\label{eq_problem} \\
\text{s.t.}~~& \liminf_{n\rightarrow \infty} \frac{1}{n}{\mathbb{E}\left[S_{n}\right]} \geq \frac{1}{f_{\max}},\label{eq_constraint}
\end{align}
where $\bar p_{\text{opt,2}}$ is the optimal value of  \eqref{eq_problem}. 
We assume that $\bar p_{\text{opt,1}}$ and $\bar p_{\text{opt,2}}$ are finite. 
The problems for maximizing the average age utility can be readily obtained from \eqref{eq_problem2} and \eqref{eq_problem} by choosing $p(\age) = - u(\age)$.  In practice, the cost for data updates increases with the average sampling rate. 
Therefore, Problems \eqref{eq_problem2} and \eqref{eq_problem} represent a tradeoff between data staleness (freshness) and update cost. 

Problems \eqref{eq_problem2} and \eqref{eq_problem} are constrained MDPs, one with a continuous (uncountable) state space and the other with a countable state space. 
Because of the  \emph{curse of dimensionality} \cite{Bellman1957}, it is quite rare that one can explicitly solve such problems and derive analytical or closed-form solutions that are arbitrarily accurate. 

\section{Main Results: Optimal Sampling Policies}\label{sec_main_results}

In this section, we present a complete characterization of the solutions to \eqref{eq_problem2} and \eqref{eq_problem}. 
Specifically, the optimal sampling policies are either deterministic or randomized threshold policies, depending on the scenario under consideration. 
Efficient computation algorithms of the thresholds and the randomization probabilities are provided. 

\subsection{Continuous-time Sampling  without Rate Constraint}

We first consider the continuous-time sampling problem \eqref{eq_problem2}. When there is no sampling rate constraint (i.e., $f_{\max} = \infty$), a solution to \eqref{eq_problem2} is provided in the following theorem:

\begin{theorem}[Continuous-time Sampling without Rate Constraint]\label{thm1}
If $f_{\max} = \infty$, $p(\cdot)$ is non-decreasing, and the service times $Y_i$ are {i.i.d.} with $0<\mathbb{E}[Y_i]<\infty$, then 
$(S_1(\beta),S_2(\beta),\ldots)$ with a parameter $\beta$ is an optimal solution to \eqref{eq_problem2}, where 
\begin{align}\label{thm1_eq12}
S_{i+1}(\beta)\! &=\!\inf\{ t\geq  D_i(\beta): \mathbb{E}[ p( \age_{t + Y_{i+1}})]\! \geq \beta \},\!\!
\end{align}
$D_i(\beta) = S_i(\beta) +Y_i$, $\age_{t} = t-S_{i}(\beta)$, and $\beta$ is the root  of 
\begin{align}\label{thm1_eq22}
\beta = \frac{\mathbb{E}\left[\int_{D_i(\beta)}^{D_{i+1}(\beta)}\! p(\age_t)dt\right]}{\mathbb{E}[D_{i+1}(\beta)\!-\!D_i(\beta)]}\!. \!\!\!\!
\end{align}
Further, $\beta$ is exactly the optimal value to \eqref{eq_problem2}, i.e., \emph{$\beta =\bar p_{\text{opt,1}}$}. 
\end{theorem} 
The proof of Theorem \ref{thm1} is relegated to  Section \ref{sec_continuous}. 
The optimal sampling policy in \eqref{thm1_eq12}-\eqref{thm1_eq22} has a nice structure. Specifically, the $(i+1)$-th sample is generated at the earliest time $t$ satisfying two conditions: \emph{(i)} the $i$-th sample has already been delivered by time $t$, i.e., $t\geq D_i(\beta)$, and \emph{(ii)} the expected age penalty $\mathbb{E}[ p( \age_{t + Y_{i+1}})]$ has grown to be no smaller than a pre-determined threshold $\beta$. Notice that if $t = S_{i+1}(\beta)$, then $t + Y_{i+1} =S_{i+1}(\beta) + Y_{i+1} = D_{i+1}(\beta)$ is the delivery time of the $(i+1)$-th sample.
In addition, $\beta$ is equal to the optimum objective value $\bar p_{\text{opt,1}}$ of \eqref{eq_problem2}. Hence, \eqref{thm1_eq12}-\eqref{thm1_eq22} require that the expected age penalty  upon the delivery of the $(i+1)$-th sample is no smaller than $\bar p_{\text{opt,1}}$, i.e., the minimum possible time-average expected age penalty.

Next, we develop an efficient algorithm to find the root $\beta$ of \eqref{thm1_eq22}. 
Because the $Y_i$'s are \emph{i.i.d.}, the expectations on the right-hand side of \eqref{thm1_eq22} are functions of $\beta$ and are irrelevant of $i$. Given $\beta$, these expectations can be evaluated by Monte Carlo simulations or importance sampling. Define
\begin{align}
v(s) = \int_0^s p(t)dt,
\end{align}
then
\begin{align}\label{eq_integral}
\int_{D_i(\beta)}^{D_{i+1}(\beta)}\!\! p(\age_t)dt\! = v(D_{i+1}(\beta)\!-\!S_i(\beta))\!-\!v(Y_i),\!\!
\end{align}
which can be used to simplify the numerical evaluation of the expected integral in \eqref{thm1_eq22}. As proven in Section \ref{sec_continuous}, 
\eqref{thm1_eq22} has a unique solution. We use a simple bisection method to solve \eqref{thm1_eq22}, which is illustrated in Algorithm \ref{alg1}. 
 



\subsubsection{Optimality Condition of  Zero-wait Sampling}
When $f_{\max} = \infty$, one logical sampling policy is the zero-wait sampling policy \cite{2015ISITYates,AgeOfInfo2016,KaulYatesGruteser-Infocom2012}, given by
\begin{align} \label{eq_Zero_wait}
S_{i+1} = S_i+ Y_i.
\end{align} 
This zero-wait sampling policy achieves the maximum throughput and the minimum queueing delay. In the special case of $p(\age_t)=\age_t$, Theorem 5 of \cite{AgeOfInfo2016} provided a sufficient and necessary condition for characterizing the optimality of the zero-wait sampling policy. We now generalize that result to the case of non-linear age functions in the following corollary:

\begin{corollary}\label{coro}
If $f_{\max} = \infty$, $p(\cdot)$ is non-decreasing, and the service times $Y_i$ are {i.i.d.} with $0<\mathbb{E}[Y_i]<\infty$, then  the zero-wait sampling policy in \eqref{eq_Zero_wait}
is optimal for solving \eqref{eq_problem2} if and only if 
\emph{
\begin{align}\label{eq_coro1}
\mathbb{E}\left[p(\text{ess}\inf Y_i +Y_{i+1})\right] \geq \frac{\mathbb{E}\left[\int_{Y_i}^{Y_i+Y_{i+1}}\! p(t)dt\right]}{\mathbb{E}[Y_{i+1}]}\!, \!\!\!\!
\end{align}}\!\!
where \emph{$\text{ess}\inf Y_i = \inf\{y\in[0,\infty): \Pr[Y_i\leq y]>0\}$}. 

\end{corollary}
\begin{proof}
See Appendix \ref{app_coro}. 
\end{proof} 

One can consider $\text{ess}\inf Y_i$ as the minimum possible value of  $Y_i$.
It immediately follows  from Corollary \ref{coro} that


\begin{corollary}\label{coro4}
If $f_{\max} = \infty$, $p(\cdot)$ is non-decreasing, and the service times $Y_i$ are {i.i.d.} with $0<\mathbb{E}[Y_i]<\infty$, then the following assertions are true:
\begin{itemize}
\item [(a).] If $Y_i$ is a constant, then  \eqref{eq_Zero_wait} is optimal for solving \eqref{eq_problem2}. 
\item [(b).] If \emph{$\text{ess}\inf Y_i = 0$} and $p(\cdot)$ is strictly increasing, then  \eqref{eq_Zero_wait} is \textbf{not} optimal for solving \eqref{eq_problem2}. 
\end{itemize}
\end{corollary}
\begin{proof}
See Appendix \ref{app_coro4}. 
\end{proof} 
\begin{algorithm}
\caption{Bisection method for solving \eqref{thm2_eq4}} \label{alg2}
\begin{algorithmic}[]
\STATE \textbf{given} $l$, $u$, tolerance $\epsilon>0$.
\REPEAT
\STATE $\beta:= (l+u)/2$.
\STATE $o_1 :=\mathbb{E} [T_{i,\min}(\beta)-S_i(\beta)]$.
\STATE $o_2 :=\mathbb{E} [T_{i,\max}(\beta)-S_i(\beta)]$.
\STATE \textbf{if} $o_1> \frac{1}{f_{\max}}$, $u:=\beta$; 
\STATE \textbf{else if} $o_2< \frac{1}{f_{\max}}$, $l:=\beta$;
\STATE \textbf{else} \textbf{return} $\beta$.
\UNTIL $u-l\leq \epsilon$.
\STATE \textbf{return} $\beta$.
\end{algorithmic}
\end{algorithm}

The condition $\text{ess}\inf Y_i = 0$ is satisfied by many commonly used distributions, such as exponential distribution, geometric distribution, Erlang distribution, and hyperexponential distribution. According to Corollary \ref{coro4}(b), if $p(\cdot)$ is strictly increasing,  the zero-wait sampling policy  \eqref{eq_Zero_wait} is not optimal for these commonly used distributions.

\subsection{Continuous-time Sampling  with Rate Constraint}

When the sampling rate constraint \eqref{eq_constraint2} is imposed, a solution to \eqref{eq_problem2} is presented in the following theorem: 

\begin{theorem}[Continuous-time Sampling with Rate Constraint]\label{thm2}
If $p(\cdot)$ is non-decreasing, $\mathbb{E} [p(t+Y_i)]<\infty$ for all  finite $t$, and the service times $Y_i$ are {i.i.d.} with $0<\mathbb{E}[Y_i]<\infty$, then 
\eqref{thm1_eq12}-\eqref{thm1_eq22} is an optimal solution to \eqref{eq_problem2}, 
if 
\begin{align}\label{thm2_eq0}
\mathbb{E}[S_{i+1}(\beta)-S_i(\beta)]>\frac{1}{f_{\max}}.
\end{align}
Otherwise, $(S_1(\beta),S_2(\beta),\ldots)$  with a parameter $\beta$ is an optimal solution to \eqref{eq_problem2}, where 
\begin{align}\label{thm2_eq1}
S_{i+1}(\beta)\! =\left\{\begin{array}{l l} T_{i,\min}(\beta)\text{ with probability }\lambda,\\ T_{i,\max}(\beta) \text{ with probability }1-\lambda, \end{array}\right.
\end{align}
$T_{i,\min}(\beta)$ and $T_{i,\max}(\beta)$ are given by
\begin{align}\label{thm2_eq2}
\!\!T_{i,\min}(\beta)\! &=\!\inf\{ t\geq  D_i(\beta): \mathbb{E}[ p( \age_{t + Y_{i+1}})]\! \geq\! \beta \},\!\!\\
\!\!T_{i,\max}(\beta)\! &=\!\inf\{ t\geq  D_i(\beta): \mathbb{E}[ p( \age_{t + Y_{i+1}})]\! >\! \beta \},\!\label{thm2_eq3}
\end{align}
 $D_i(\beta) = S_i(\beta) +Y_i$, $\age_{t} = t-S_{i}(\beta)$, $\beta$ is determined by solving
\begin{align}\label{thm2_eq4}
\mathbb{E} [T_{i,\min}(\beta)-S_i(\beta)] \leq \frac{1}{f_{\max}}\leq \mathbb{E} [T_{i,\max}(\beta)-S_i(\beta)],\!\!
\end{align}
and $\lambda$ is given by\footnote{If $T_{i,\min}(\beta) = T_{i,\max}(\beta)$ almost surely,  then $\eqref{thm2_eq1}$ becomes a deterministic threshold policy and  $\lambda$ can be any number within $[0,1]$.}
\begin{align}\label{thm2_eq5}
\lambda = \frac{\mathbb{E} [T_{i,\max}(\beta)-S_i(\beta)] -\frac{1}{f_{\max}}  }{\mathbb{E} [T_{i,\max}(\beta)-T_{i,\min}(\beta)]}.
\end{align}
\end{theorem} 
The proof of Theorem \ref{thm2} will be provided in  Section \ref{sec_analysis}. 
According to Theorem \ref{thm2}, the solution to \eqref{eq_problem2} consists of two cases: In \emph{Case 1},  the deterministic threshold policy in Theorem \ref{thm1} is an optimal solution to \eqref{eq_problem2}, which needs to satisfy  \eqref{thm2_eq0}. 
In \emph{Case 2},  the randomized threshold policy in \eqref{thm2_eq1}-\eqref{thm2_eq5} is an optimal solution to \eqref{eq_problem2}, which needs to satisfy 
\begin{align}\label{eq_equality_constraint}
\mathbb{E} [S_{i+1}(\beta)-S_i(\beta)] = \frac{1}{f_{\max}}.
\end{align}
We note that the only difference between 
\eqref{thm2_eq2} and \eqref{thm2_eq3} is that ``$\geq$'' is used 
in \eqref{thm2_eq2} while 
``$>$'' is employed 
 in \eqref{thm2_eq3}.\footnote{Clearly, an important issue is the optimality of such a randomized threshold policy, which is proven in Section \ref{sec_analysis}.} If 
 there exists a time-interval $[a,b]$ such that 
\begin{align}
\mathbb{E}[ p( {t + Y_{i+1}})] = \beta \text{ for all } t\in[a,b], 
\end{align}
as shown in Fig. \ref{fig_solution}(a), then $T_{i,\min}(\beta)<  T_{i,\max}(\beta)$. 
In this case, the choices $S_{i+1}(\beta)=T_{i,\min}(\beta)$ and $S_{i+1}(\beta)=T_{i,\max}(\beta)$ may not satisfy \eqref{eq_equality_constraint}, but their randomized mixture in \eqref{thm2_eq1} can satisfy \eqref{eq_equality_constraint}. In particular, if $\beta$ and $\lambda$ are given by  \eqref{thm2_eq4} and \eqref{thm2_eq5}, then 
\eqref{eq_equality_constraint} is satisfied.
\begin{figure}
\centering
\includegraphics[width=0.48\textwidth]{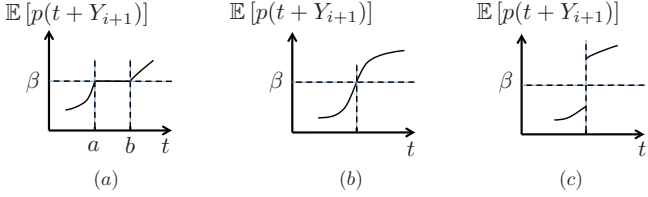}   
\caption{Three cases of function $f(t) = \mathbb{E}\left[ p(t+Y_{i+1})\right]$.}
\label{fig_solution}
\end{figure}

We provide a low-complexity  algorithm to compute the randomized threshold policy in \eqref{thm2_eq1}-\eqref{thm2_eq5}:  
As shown in Appendix \ref{app_thm_strong_duality}, there is a unique $\beta$ satisfying \eqref{thm2_eq4}. We use the bisection method in Algorithm \ref{alg2} to solve \eqref{thm2_eq4} and obtain $\beta$. 
After that, $S_{i+1}(\beta)$ and $\lambda$  can be computed by substituting  $\beta$ into \eqref{thm2_eq1}-\eqref{thm2_eq3} and \eqref{thm2_eq5}. 
Because of the similarity between \eqref{thm2_eq2} and \eqref{thm2_eq3}, $S_{i+1}(\beta)$ and $\lambda$ 
are quite sensitive to the numerical error in $\beta$. This issue can be resolved by replacing $T_{i,\min}(\beta)$ in \eqref{thm2_eq1} and \eqref{thm2_eq5} with $T_{i,\min}'(\beta)$ and replacing $T_{i,\max}(\beta)$ in \eqref{thm2_eq1} and \eqref{thm2_eq5} with $T_{i,\max}'(\beta)$, where $T_{i,\min}'(\beta)$ and $T_{i,\max}'(\beta)$ are determined by
\begin{align}\label{}
\!\!T_{i,\min}'(\beta)\! &=\!\inf\{ t\geq  D_i(\beta): \mathbb{E}[ p( \age_{t + Y_{i+1}})]\! \geq\! \beta-\epsilon/2 \},\!\!\\
\!\!T_{i,\max}'(\beta)\! &=\!\inf\{ t\geq  D_i(\beta): \mathbb{E}[ p( \age_{t + Y_{i+1}})]\! >\! \beta +\epsilon/2\},\!\label{}
\end{align}
respectively, and $\epsilon>0$ is the tolerance in Algorithm \ref{alg2}. 
One can improve the accuracy of this solution  by \emph{(i)} reducing the tolerance $\epsilon$ and \emph{(ii)} computing the expectations more accurately  by increasing the number of Monte Carlo realizations or using advanced techniques such as importance sampling. 



As depicted in Fig. \ref{fig_solution}(b)-(c), if $\mathbb{E} [p(t+Y_{i+1})] $ is strictly increasing on $t\in[0,\infty)$, then 
$T_{i,\min}(\beta) = T_{i,\max}(\beta)$ almost surely and \eqref{thm2_eq1} reduces to a deterministic threshold policy. In this case, Theorem \ref{thm2} can be greatly simplified, as 
stated in the following corollary: 

\begin{corollary}\label{coro2}
In Theorem \ref{thm2}, if $\mathbb{E} [p(t+Y_{i+1})] $ is strictly increasing in $t$, then 
\eqref{thm1_eq12} 
 is an optimal solution to \eqref{eq_problem2}, where $D_i(\beta) = S_i(\beta) +Y_i$, $\age_{t} = t-S_{i}(\beta)$, and 
 $\beta$ is determined by \eqref{thm1_eq22}, if 
\begin{align}
\mathbb{E}[S_{i+1}(\beta)-S_i(\beta)]>\frac{1}{f_{\max}};
\end{align}
otherwise, $\beta$ is determined by solving 
\begin{align}
\mathbb{E}[S_{i+1}(\beta)-S_i(\beta)]=\frac{1}{f_{\max}}. 
\end{align}
 \end{corollary}
The proof of Corollary \ref{coro2} is omitted, because it  follows immediately from Theorem \ref{thm2}. 

If $p(\cdot)$ is strictly increasing or the distribution of $Y_i$ is sufficiently smooth, $\mathbb{E} [p(t+Y_{i+1})] $ is strictly increasing in $t$. Hence, the extra condition in Corollary \ref{coro2} is satisfied for a broad class of age penalty functions and service time distributions. 
%

A restrictive case of problem \eqref{eq_problem2} was  studied in \cite{AgeOfInfo2016}, where $p(\cdot)$ was assumed to be positive and non-decreasing. There is an error in Theorem 3 of \cite{AgeOfInfo2016}, because the condition ``$\mathbb{E} [p(t+Y_{i+1})] $ is strictly increasing in $t$" is missing. Further, the solution in Theorem 3 of \cite{AgeOfInfo2016} is more complicated than that in Corollary \ref{coro2}. A special case of Corollary \ref{coro2} with $p(t) = t$ was derived in Theorem 4 of \cite{AgeOfInfo2016}.

%

%


\subsection{Discrete-time Sampling}
We now move on to the discrete-time sampling problem \eqref{eq_problem}. When
there is no  sampling rate constraint (i.e., $f_{\max} = \infty$), the solution to \eqref{eq_problem} is provided in the following theorem:

\begin{theorem}[Discrete-time Sampling without Rate Constraint]\label{thm3}
If $f_{\max} = \infty$, $p(\cdot)$ is non-decreasing, and the service times $Y_i$ are {i.i.d.} with $0<\mathbb{E}[Y_i]<\infty$, then 
$(S_1(\beta),S_2(\beta),\ldots)$ is an optimal solution to \eqref{eq_problem}, where 
\begin{align}\label{thm3_eq12}
\!\!\!S_{i+1}(\beta)\! &=\!\min\{ t\in \mathbb{N}: t\geq  D_i(\beta), \mathbb{E}\left[p(\age_{t+Y_{i+1}})\right]\! \geq \beta \},\!\!
\end{align}
$D_i(\beta) = S_i(\beta) +Y_i$, $\age_{t} = t-S_{i}(\beta)$, and $\beta$ is the root  of 
\begin{align}\label{thm3_eq22}
\beta = \frac{\mathbb{E}\left[\sum_{t =D_i(\beta)}^{D_{i+1}(\beta)-1}\! p(\age_t)\right]}{\mathbb{E}[D_{i+1}(\beta)\!-\!D_i(\beta)]}\!. \!\!\!\!
\end{align}
Further, $\beta$ is exactly the optimal value to \eqref{eq_problem}, i.e., \emph{$\beta =\bar p_{\text{opt,2}}$}. 
\end{theorem} 
The proofs of the discrete-time sampling results will be discussed in  Section \ref{sec_discrete_proofs}. 
Theorem \ref{thm3} is quite similar to Theorem \ref{thm1}, with two minor differences: \emph{(i)} The  sampling time $S_{i+1}(\beta)$ in \eqref{thm1_eq12} is a real number, which is restricted to an integer in \eqref{thm3_eq12}. \emph{(ii)} The integral in \eqref{thm1_eq22} becomes a summation in \eqref{thm3_eq22}. 

In the discrete-time case, the optimality of the zero-wait sampling policy is characterized as  follows.

\begin{corollary}\label{coro3}
If $f_{\max} = \infty$, $p(\cdot)$ is non-decreasing, and the service times $Y_i$ are {i.i.d.} with $0<\mathbb{E}[Y_i]<\infty$, then the zero-wait sampling policy \eqref{eq_Zero_wait} is optimal for solving \eqref{eq_problem} if and only if there exists $e<1$ such that
\emph{
\begin{align}\label{}
\mathbb{E}\left[p(\text{ess}\inf Y_i +Y_{i+1} + e)\right] \geq \frac{\mathbb{E}\left[\sum_{t=Y_i}^{Y_i+Y_{i+1}-1}\! p(t)\right]}{\mathbb{E}[Y_{i+1}]}\!, \!\!\!\!
\end{align}}\!\!
where \emph{$\text{ess}\inf Y_i = \min\{y\in\mathbb{N}: \Pr[Y_i\leq y]>0\}$}. 
\end{corollary}

When the sampling rate constraint \eqref{eq_constraint} is imposed,  the solution to \eqref{eq_problem} is provided in the following theorem.

\begin{theorem}[Discrete-time Sampling with Rate Constraint]\label{thm4}
If $p(\cdot)$ is non-decreasing, $\mathbb{E} [p(t+Y_i)]<\infty$ for all  finite $t$, and the service times $Y_i$ are {i.i.d.} with $0<\mathbb{E}[Y_i]<\infty$, then 
\eqref{thm3_eq12}-\eqref{thm3_eq22}  is an optimal solution to \eqref{eq_problem}, if
\begin{align}\label{thm4_eq0}
\mathbb{E}[S_{i+1}(\beta)-S_i(\beta)]> \frac{1}{f_{\max}}.
\end{align}
Otherwise, $(S_1(\beta),S_2(\beta),\ldots)$ is an optimal solution to \eqref{eq_problem}, where 
\begin{align}\label{thm4_eq1}
S_{i+1}(\beta)\! =\left\{\begin{array}{l l} T_{i,\min}(\beta)\text{ with probability }\lambda,\\ T_{i,\max}(\beta) \text{ with probability }1-\lambda, \end{array}\right.
\end{align}
$T_{i,\min}(\beta)$ and $T_{i,\max}(\beta)$ are given by
\begin{align}\label{thm4_eq2}
\!\!\!\!T_{i,\min}(\beta)\! &=\!\min\{ t\in \mathbb{N}\!:\! t\geq  D_i(\beta), \mathbb{E}\left[p(\age_{t+Y_{i+1}})\right]\! \geq\! \beta \},\!\!\\
\!\!\!\!T_{i,\max}(\beta)\! &=\!\min\{ t\in \mathbb{N}\!:\! t\geq  D_i(\beta),   \mathbb{E}\left[p(\age_{t+Y_{i+1}})\right]\! >\! \beta \},\!\!\!\label{thm4_eq3}
\end{align}
 $D_i(\beta) = S_i(\beta) +Y_i$, $\age_{t} = t-S_{i}(\beta)$, $\beta$ is determined by solving
\begin{align}\label{thm4_eq4}
\!\!\mathbb{E} [T_{i,\min}(\beta)-S_i(\beta)] \leq \frac{1}{f_{\max}}\leq \mathbb{E} [T_{i,\max}(\beta)-S_i(\beta)],\!\!
\end{align}
and $\lambda$ is given by
\begin{align}\label{thm4_eq5}
\lambda = \frac{\mathbb{E} [T_{i,\max}(\beta)-S_i(\beta)] -\frac{1}{f_{\max}}  }{\mathbb{E} [T_{i,\max}(\beta)-T_{i,\min}(\beta)]}.
\end{align}
\end{theorem} 

Theorem \ref{thm4} is similar to Theorem \ref{thm2}, but there are two differences: \emph{(i)} $T_{i,\min}(\beta)$ and $T_{i,\max}(\beta)$ are real numbers in \eqref{thm2_eq2}-\eqref{thm2_eq3}, which are restricted to  integers in \eqref{thm4_eq2}-\eqref{thm4_eq3}. \emph{(ii)} If $\mathbb{E} [p(t+Y_{i+1})] $ is strictly increasing in $t$, then $T_{i,\min}(\beta) = T_{i,\max}(\beta)$ holds almost surely in \eqref{thm2_eq2}-\eqref{thm2_eq3} and Theorem \ref{thm2} can be greatly simplified. However, in the discrete-time case, even if $\mathbb{E} [p(t+Y_{i+1})] $ is strictly increasing in $t$, $T_{i,\min}(\beta)<T_{i,\max}(\beta)$ may still occur in \eqref{thm4_eq2}-\eqref{thm4_eq3}. In fact, it is rather common that $T_{i,\min}(\beta)<T_{i,\max}(\beta)$ holds for the optimal $\beta$, because of the following  reason:  

If $T_{i,\min}(\beta) = T_{i,\max}(\beta)$ almost surely, then \eqref{thm4_eq1} becomes  a deterministic threshold policy that needs to ensure \eqref{eq_equality_constraint}. 
 However, because $S_{i+1}(\beta)$ and $S_{i}(\beta)$ are integers, such a deterministic threshold policy is difficult to satisfy \eqref{eq_equality_constraint} for certain values of $f_{\max}$. On the other hand, if $T_{i,\min}(\beta)< T_{i,\max}(\beta)$, the randomized threshold policy in \eqref{thm4_eq1}-\eqref{thm4_eq5} can satisfy \eqref{eq_equality_constraint}. Hence, even though $\mathbb{E} [p(t+Y_{i+1})] $ is strictly increasing in $t$, Theorem \ref{thm4} cannot be further simplified. This is a key difference between continuous-time and discrete-time sampling. 

The computation algorithms of the optimal discrete-time sampling policies are similar to their counterparts in the continuous-time case, and hence are omitted. 

\subsection{An Example: Mutual Information Maximization}
Next, we provide an example to illustrate the above theoretical results. Suppose that $X_t$ is a stationary, time-homogeneous Markov chain and the sampling times $S_i$ are independent of   $\{X_t,t\geq 0\}$. The optimal sampling problem that maximizes the time-average expected  mutual information between $X_t$ and $\bm{W}_t$ is formulated as
\vspace{-0mm}
\begin{equation}\label{eq_problem1}
\bar I_{\text{opt}}=\sup_{\pi\in\Pi}~\liminf_{n\rightarrow \infty}\frac{1}{n}~\mathbb{E}\left[\sum_{t=1}^nI(X_t; \bm{W}_t)\right],
\vspace{-0mm}
\end{equation}
where $\bar I_{\text{opt}}$ is the optimal value of \eqref{eq_problem1}. We assume that $\bar I_{\text{opt}}$ is finite. Problem \eqref{eq_problem1} is a special case of \eqref{eq_problem} satisfying  $p(\age_t) = -u(\age_t)= -I(X_t; \bm{W}_t)$ and $f_{\max} = \infty$. The following result follows immediately from Theorem \ref{thm3}.  

\begin{corollary}\label{coro5}
If the service times $Y_i$ are {i.i.d.} with $0<\mathbb{E}[Y_i]<\infty$, then $(S_1(\beta),S_2(\beta),\ldots)$ is an optimal solution to \eqref{eq_problem}, where
\begin{align}\label{coro5_eq1}
S_{i+1}(\beta)
& \!=\! \min\{t\!\in\! \mathbb{N}\!:\! t \geq D_i(\beta), \nonumber\\ 
&~~~~~~~~~~~~~~~~~~~I(X_{t+Y_{i+1}}; X_{S_i(\beta)}| Y_{i+1})\!\leq\! \beta \},\!\!\!
\end{align}
$D_i(\beta) = S_i(\beta) +Y_i$,  and $\beta\geq 0$ is the root  of 
\emph{
\begin{align}\label{coro5_eq2}
\beta=\frac{\mathbb{E}\bigg[\sum\limits_{t=D_i(\beta)}^{D_{i+1}(\beta)-1} I(X_t;X_{S_i(\beta)})\bigg]}{\mathbb{E}[D_{i+1}(\beta)-D_i(\beta)]}.
\end{align}}
\!\!Further, $\beta$ is exactly the optimal value of \eqref{eq_problem1}, i.e., \emph{$\beta=\bar I_{\text{opt}}$}.
\end{corollary}

{In Corollary \ref{coro5}, the next sampling time $S_{i+1}(\beta)$ is determined based on the mutual information between the freshest received sample ${X}_{S_i(\beta)}$ and the source value ${X}_{D_{i+1}(\beta)}$, where $D_{i+1}(\beta)=S_{i+1}(\beta)+Y_{i+1}$ is the delivery time of the $(i+1)$-th sample. Because $Y_{i+1}$ will be known by both the transmitter and receiver at time $D_{i+1}(\beta)=S_{i+1}(\beta)+Y_{i+1}$, $Y_{i+1}$ is the side information in the conditional mutual information $I[X_{t+Y_{i+1}}; {X}_{S_i(\beta)} | Y_{i+1} ]$. The conditional mutual information $I[X_{t+Y_{i+1}}; {X}_{S_i(\beta)} | Y_{i+1} ]$ decreases as time $t$ grows. According to \eqref{coro5_eq1}, the $(i+1)$-th sample is generated at the smallest integer $t$ satisfying two conditions: \emph{(i)} the $i$-th sample has already been delivered by time $t$ and \emph{(ii)} the conditional mutual information $I[X_{t+Y_{i+1}}; {X}_{S_i(\beta)} | Y_{i+1} ]$ has reduced to be no greater than $\bar I_{\text{opt}}$, i.e., the optimum of the time-average expected mutual information $\liminf_{n\rightarrow \infty}\frac{1}{n}~\mathbb{E}\left[\sum_{t=1}^nI(X_t; \bm{W}_t)\right]$ that we are maximizing.} 

The optimal sampling policy  is illustrated in Fig. \ref{fig_policy}, where $\beta = 0.1$ and $Y_i$ is equal to either $1$ or $5$  with equal probability. The sampling time $S_{i}(\beta)$, delivery time $D_{i}(\beta)$, and  conditional mutual information $I[X_{t+Y_{i+1}}; {X}_{S_i(\beta)} | Y_{i+1} ]$ are depicted in the figure. One can observe that if the service time of the previous sample is $Y_i = 1$, the sampler will wait until the conditional mutual information drops below the threshold $\beta$ and then take the next sample; if the service time of the previous sample is $Y_i =5$,  the next sample is taken once the previous sample is delivered, because the conditional mutual information is already below $\beta$ then.

\begin{figure}
\centering
\includegraphics[width=0.35\textwidth]{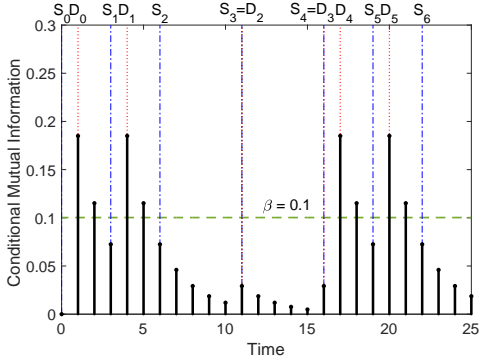}   
\caption{A sample-path illustration of the optimal sampling policy \eqref{coro5_eq1} and \eqref{coro5_eq2}, where $\beta = 0.1$, $Y_i$ is either $1$ or $5$ with equal probability, $S_i$ and $D_i$ are sampling time and delivery time of the $i$-th sample. 
On this sample-path, the service times are $Y_0=1, Y_1 = 1, Y_2 = 5, Y_3 = 5, Y_4 = 1, Y_5 = 1, Y_6 = 5$.  }
\label{fig_policy}
\end{figure}    

\subsection{Alternative Expressions of the Threshold Sampling Policy}
Finally, we  present two alternative expressions of the sampling policy \eqref{thm1_eq12}. 
Define 
\begin{align}
w(\beta) = \!\inf\{ \age \geq 0 : \mathbb{E}[ p( \age + Y_{i+1})]\! \geq \beta \},
\end{align}
then \eqref{thm1_eq12} can be rewritten as 
\begin{align}\label{eq_threshold_age}
S_{i+1}(\beta)\! &=\!\inf\{ t\geq  D_i(\beta):  \age_{t} \geq w(\beta) \},
\end{align}
which is a threshold policy on the age $\age_{t}$. Threshold policies similar to \eqref{eq_threshold_age} were 
discussed in age minimization for status update systems with an energy harvesting constraint, e.g., \cite{Bacinoglu2015,Bacinoglu2017,Wu2018,ArafaArXiv2018,BacinogluISIT2018}. The technical tools therein are significantly different from ours, because of the energy harvesting constraint. 
Further, from \eqref{eq_age2} and \eqref{eq_threshold_age}, we get 
\begin{align}\label{eq_threshold_age1}
S_{i+1}(\beta)\! &=\!\inf\{ t\geq  D_i(\beta):  \age_{t} \geq w(\beta) \},\!\! \nonumber \\
&=\!\inf\{ t\geq  D_i(\beta):  t  \geq w(\beta) + D_i(\beta) - Y_i \},\!\! \nonumber\\
& = D_i(\beta) + \max\{w(\beta) - Y_i, 0\}. 
\end{align}
We use $Z_i(\beta) = S_{i+1}(\beta)-D_i(\beta)\geq 0$ to denote the waiting time from the delivery time of the $i$-th sample to the generation time of the $(i+1)$-th sample. 
By \eqref{eq_threshold_age1}, $Z_i(\beta)$ can be expressed as a simple water-filling solution, i.e.,
\begin{align}
Z_i(\beta) = \max\{w(\beta) - Y_i, 0\}, 
\end{align}
where $w(\beta)$ is the water level. Hence, the  waiting time $Z_i(\beta)$ decreases linearly with the service time $Y_i$, until $Z_i(\beta)$ drops to zero. 
The water-filling solution was shown to be age-optimal for a special case that  $p(\age_t) = \age_t$  \cite{2015ISITYates,AgeOfInfo2016}. Recently, it was observed via simulations that the water-filling solution comes very close to the optimal age performance in symmetric multi-source networks \cite{AhmedMobiHoc2019}.

\section{Proofs of the Main Results}\label{sec_analysis}
In this section, we prove the main results in Section \ref{sec_main_results}, by using  the technical tools recently developed  in \cite{SunISIT2017,SunTIT2018}, as well as some additional proof ideas that are needed for showing Lemma \ref{lem_optimal_eq_opt_stopping22}, Theorem \ref{thm6_strong_duality}, and Lemma \ref{lem_optimal_eq_opt_stopping23} below.

We begin with the proof of Theorem \ref{thm2}, because its proof procedure is helpful for presenting and understanding the other proofs.




%

\subsection{Suspend Sampling when the Server is Busy}\label{sec_proof_1}
In \cite{AgeOfInfo2016}, it was shown that \emph{no new sample should be taken when the server is busy}. The reason is as follows: If a sample is taken when the server is busy, it has to wait in the queue for its transmission opportunity, during which time the sample is becoming stale. A better strategy is to take a new sample just when the server becomes idle, which yields a smaller age process on sample path. This comparison leads to the following lemma: 
\begin{lemma}\label{lem_zeroqueue}
 In the optimal sampling problem \eqref{eq_problem2}, 
it is suboptimal to take a new sample before the previous sample is delivered.
\end{lemma}

By Lemma \ref{lem_zeroqueue}, the queue in Figure \ref{fig_model} should be always kept empty. In addition, we only need to consider a sub-class of sampling policies $\Pi_1\subset\Pi$ in which each sample is generated after the previous sample is delivered, i.e.,
\begin{align}
\Pi_1 
 &= \{\pi\in\Pi: S_{i+1} \geq D_{i} = S_i + Y_i \text{ for all $i$}\}. 
\end{align}
Let $Z_i = S_{i+1} -D_{i} \geq0$ represent the waiting time between the delivery time $D_i$ of the $i$-th sample and the generation time $S_{i+1}$ of the $(i+1)$-th sample. Since $S_0 = 0$, we have
$S_i =  S_{0} +\sum_{j=0}^{i} (Y_{j} + Z_j)=\sum_{j=0}^{i} (Y_{j} + Z_j)$ and $D_i = S_i + Y_i$. Given $(Y_1,Y_2,\ldots)$, $(S_1,S_2,\ldots)$ is uniquely determined by $(Z_1,Z_2,\ldots)$. Hence, one can also use $\pi = (Z_1,Z_2,\ldots)$ to represent a sampling policy in $\Pi_1$.

Because $T_i $ is a regenerative process, 
by using the renewal theory in \cite{Ross1996} and \cite[Section 6.1]{Haas2002}, one can show that $\frac{1}{i} 
\mathbb{E}[S_i]$ and $\frac{1}{i} 
\mathbb{E}[D_i]$ in  \eqref{eq_problem2} are convergent sequences and 
\begin{align}
&\limsup_{T\rightarrow \infty}\frac{1}{T}~\mathbb{E}\left[\int_{0}^T p(\age_t) dt\right] \nonumber
\\
=& \lim_{i\rightarrow \infty}\frac{\mathbb{E}\left[\int_{0}^{D_i} p(\age_t) dt\right]}{\mathbb{E}[D_i]} \nonumber\\
=& \lim_{i\rightarrow \infty}\frac{\sum_{j=1}^{i}\mathbb{E}\left[\int_{D_{j}}^{D_{j+1}-1}p(\age_t) dt\right]}{\sum_{j=1}^{i} \mathbb{E}\left[Y_j+Z_j\right]}.
\end{align}
On the other hand, it follows from \eqref{eq_age2} that 
\begin{align}\label{eq_sum1}
\int_{D_{j}}^{D_{j+1}}p(\age_t) dt&=\!\!\int_{D_{j}}^{D_{j+1}}p(t-S_i) dt\nonumber\\
&=\!\! \int_{Y_{i}}^{Y_i+Z_i+Y_{i+1}}p(t) dt, 
\end{align}
which is a function of $(Y_i,Z_i,Y_{i+1})$. Define 
\begin{align}\label{eq_sum}
q(y_i,z,y')= \int_{y}^{y+z+y'}p(t) dt, 
\end{align}
then Problem \eqref{eq_problem2} can be rewritten as 
\begin{align}\label{eq_problem_S11}
\bar p_{\text{opt}_1} = \inf_{\pi\in\Pi_1}~&\lim_{i\rightarrow \infty}\frac{\sum_{j=1}^{i}\mathbb{E}\left[q(Y_j,Z_j,Y_{j+1})\right]}{\sum_{j=1}^{i} \mathbb{E}\left[Y_j+Z_j\right]}
\\
\text{s.t.}~~& \lim_{i\rightarrow \infty} \frac{1}{i}\sum_{j=1}^{i} \mathbb{E}\left[Y_j+Z_j\right] \geq \frac{1}{f_{\max}}.\label{eq_problem_S111}
\end{align}
\subsection{Reformulation of Problem \eqref{eq_problem_S11}}

In order to solve \eqref{eq_problem_S11}, we consider the following MDP with a parameter $c\geq 0$:
\begin{align}\label{eq_SD1}
\!\!h(c)\!\triangleq\!\inf_{\pi\in\Pi_1}&\lim_{i\rightarrow \infty}\frac{1}{i}\sum_{j=0}^{i-1}\mathbb{E}\left[q(Y_j,Z_j,Y_{j+1})-c(Y_j+Z_j)\right]\!\!\!\!\\
\text{s.t.}~&\lim_{i\rightarrow \infty} \frac{1}{i}\sum_{j=1}^{i} \mathbb{E}\left[Y_j+Z_j\right] \geq \frac{1}{f_{\max}},\label{eq_SD2}
\end{align}
where $h(c)$ is the optimum  value of \eqref{eq_SD1}. 
Similar with Dinkelbach's method  \cite{Dinkelbach67} for nonlinear fractional programming, the following lemma holds for the MDP \eqref{eq_problem_S11}:

\begin{lemma} \label{lem_ratio_to_minus} \cite[Lemma 2]{SunTIT2018}
The following assertions are true:
\begin{itemize}
\item[(a).] \emph{$\bar p_{\text{opt}_1}  \gtreqqless c $} if and only if $h(c)\gtreqqless 0$. 
\item[(b).] If $h(c)=0$, the solutions to \eqref{eq_problem_S11}
and \eqref{eq_SD1} are identical. 
\end{itemize}
\end{lemma}

Hence, the solution to \eqref{eq_problem_S11} can be obtained by solving \eqref{eq_SD1} and seeking $\bar p_{\text{opt}_1}\in\mathbb{R}$  that satisfies
\begin{align}\label{eq_c}
h(\bar p_{\text{opt}_1})=0. 
\end{align}
\subsection{Lagrangian Dual Problem of  \eqref{eq_SD1} when \emph{$c = \bar p_{\text{opt}_1}$} }

Although \eqref{eq_SD1} is a continuous-time MDP with a continuous state space, rather than a convex optimization problem, it is possible to use the Lagrangian dual approach to solve \eqref{eq_SD1} and show that it admits no duality gap. 

When $c = \bar p_{\text{opt}_1}$, define the following Lagrangian
\begin{align}\label{eq_Lagrangian}
 L(\pi;\alpha) \! =\! & \lim_{n\rightarrow \infty}\!\frac{1}{n}\sum_{i=0}^{n-1}\mathbb{E}\left[q(Y_j,Z_j,Y_{j+1})\!-\! (\bar p_{\text{opt}_1}\! +\! \alpha)(Y_i\!+\!Z_i)\right] \nonumber\\
&~+ \frac{\alpha}{f_{\max}},
\end{align}
where $\alpha\geq0$ is the dual variable. 
Let
\begin{align}\label{eq_primal}
g(\alpha) \triangleq \inf_{\pi\in\Pi_1}  L(\pi;\alpha).
\end{align}
Then, the Lagrangian dual problem of \eqref{eq_SD1} is defined by
\begin{align}\label{eq_dual}
d \triangleq \max_{\alpha\geq 0}g(\alpha),
\end{align}
where $d$ is the optimum value of \eqref{eq_dual}. 
Weak duality \cite{Bertsekas2003,Boyd04} implies that $d \leq h(\bar p_{\text{opt}_1})$. 
Next, we will solve  \eqref{eq_primal} and establish strong duality, i.e., $d = h(\bar p_{\text{opt}_1})$.

\subsection{Optimal Solutions to \eqref{eq_primal}}

%



We solve \eqref{eq_primal} in two steps: First, we  use a sufficient statistic argument to show that \eqref{eq_primal} can be decomposed into a series of per-sample optimization problems. Second, each per-sample optimization problem is reformulated as a convex optimization problem, which is solved in closed-form. The details are provided as follows.

\begin{lemma}\label{lem_SRoptimal} 
If the service times $Y_i$ are {i.i.d.}, then $Y_i$ is a sufficient statistic for determining the optimal $Z_i$ in \eqref{eq_primal}. 
\end{lemma}
\begin{proof}
In \eqref{eq_primal}, the minimization of the term 
\begin{align}\label{eq_opt_stopping1}
&\mathbb{E}\left[q(Y_i,Z_i,Y_{i+1})- (\bar p_{\text{opt}_1}+\alpha) (Y_i+Z_i)\right] \nonumber\\
\overset{(a)}{=}&\mathbb{E}\left[q(Y_i,Z_i,Y_{i+1})- (\bar p_{\text{opt}_1}+\alpha) (Z_i+Y_{i+1})\right]
\end{align}
over $Z_i$ depends on $(Y_1,\ldots,Y_i, Z_1,\ldots,Z_{i-1})$ via $Y_i$, where  Step (a) is due to  $\mathbb{E}[Y_i]=\mathbb{E}[Y_{i+1}]$. 
Hence, $Y_i$ is a sufficient statistic for determining $Z_i$ in \eqref{eq_primal}. 
\end{proof}
By Lemma \ref{lem_SRoptimal}, \eqref{eq_primal} can be decomposed into a series of per-sample optimization problems. 
In particular, after observing the realization $Y_i=y_i$, $Z_i$ is determined by solving
\begin{align}\label{eq_opt_stopping21}
\!\!\min_{\substack{\Pr[Z_i\in A| Y_i=y_i]\\Z_i\geq 0}}\!\mathbb{E}\left[q(y_i,Z_i,Y_{i+1})\!-\! (\bar p_{\text{opt}_1}\!+\!\alpha) (Z_i\!+\!Y_{i+1}) \right]\!,\!\!
\end{align}
where the rule for determining $Z_i$ is represented by $\Pr[Z_i\in A| Y_i=y_i]$, i.e., the conditional  distribution of $Z_i$ given the occurrence of $Y_i=y_i$.
To find all possible solutions to \eqref{eq_opt_stopping21}, let us consider the following problem 
\begin{align}\label{eq_opt_stopping22}
\!\!\min_{\substack{z\geq 0}}\mathbb{E}\left[q(y_i,z,Y_{i+1})\!-\! (\bar p_{\text{opt}_1}\!+\!\alpha) (z\!+\!Y_{i+1}) \right]\!.\!\!
\end{align}
Because $p(\cdot)$ is non-decreasing,  the functions $z \rightarrow q(y_i,z,y')$ and $z \rightarrow \mathbb{E}\left[q(y_i,z,Y_{i+1})\right]$ are both convex. Hence, \eqref{eq_opt_stopping22} is a convex optimization problem. 
\begin{lemma}\label{lem_optimal_eq_opt_stopping22} 
If $p(\cdot)$ is non-decreasing, then the set of optimal solution to \eqref{eq_opt_stopping22} is $[z_{\min}(y_i,\alpha),z_{\max}(y_i,\alpha)]$ where  
\emph{
\begin{align}\label{eq_gamma}
z_{\min}(y,\alpha) &= \!\inf\{ t\geq 0: \mathbb{E}[ p(y+t + Y_{i+1}]\! \geq\! \bar p_{\text{opt}_1}\!+\!\alpha \},\!\! \\
z_{\max} (y,\alpha) 
&=\!\inf\{ t\geq 0: \mathbb{E}[ p(y+t + Y_{i+1}]\! >\! \bar p_{\text{opt}_1}\!+\!\alpha \}.\!\! \label{eq_gamma1}
\end{align}}
\end{lemma}

\begin{proof}
See Appendix \ref{app_lem_optimal_eq_opt_stopping22}.
\end{proof}
By Lemma \ref{lem_optimal_eq_opt_stopping22}, $z$ is an optimal solution to \eqref{eq_opt_stopping22} if and only if 
$z \in [z_{\min}(y_i,\alpha),z_{\max}(y_i,\alpha)]$. Hence, 
 the set of optimal solutions to \eqref{eq_opt_stopping21} is
\begin{align}
\{\Pr[Z_i\!\in\! A| Y_i\!=\!y_i]\!:\! Z_i \in [z_{\min}(y_i,\alpha),z_{\max}(y_i,\alpha)] ~~\nonumber\\
~~~\text{ almost surely}\}.\nonumber
\end{align}
Combining this with Lemma \ref{lem_SRoptimal}, yields  
\begin{lemma}\label{lem_gamma}
If $p(\cdot)$ is non-decreasing and the service times $Y_i$ are {i.i.d.} with $0<\mathbb{E}[Y_i]<\infty$, then the set of optimal solutions to \eqref{eq_primal} is 
\begin{align}\label{eq_lem_gamma}
\Gamma(\alpha) = \{\pi: Z_i \in [z_{\min}(Y_i,\alpha),z_{\max}(Y_i,\alpha)] \text{ for almost all }i\},
\end{align}
where $z_{\min}(y,\alpha)$ and $z_{\max}(y,\alpha)$ are given in \eqref{eq_gamma}-\eqref{eq_gamma1}.
\end{lemma}

\subsection{Zero Duality Gap and Optimal Solution to \eqref{eq_SD1}}\label{sec_proof_2}
Strong duality and an optimal solution to \eqref{eq_SD1} are obtained  in the following theorem:
\begin{theorem}\label{thm6_strong_duality}
If \emph{$c = \bar p_{\text{opt}_1}$},  $p(\cdot)$ is non-decreasing, $\mathbb{E} [p(t+Y_i)]<\infty$ for all  finite $t$, and the service times $Y_i$ are {i.i.d.} with $0<\mathbb{E}[Y_i]<\infty$, then the duality gap between \eqref{eq_SD1} and \eqref{eq_dual} is zero. Further, $Z_i = z_{\min}(Y_i,0)$ is an optimal solution to \eqref{eq_SD1} and \eqref{eq_dual},
if
\begin{align}\label{thm6_eq0}
\mathbb{E}[Y_i+z_{\min}(Y_i,0)]> \frac{1}{f_{\max}}.
\end{align}
Otherwise, $(Z_1,Z_2,\ldots)$ is an optimal solution to \eqref{eq_SD1} and \eqref{eq_dual}, where 
\begin{align}
Z_i = \left\{\begin{array}{l l} z_{\min}(Y_i,\alpha)\text{ with probability }\lambda,\\ z_{\max}(Y_i,\alpha) \text{ with probability }1-\lambda, \end{array}\right.
\end{align}
$\alpha\geq 0$ is determined by solving
\begin{align}
\!\!\mathbb{E}\left[Y_i+z_{\min}(Y_i,\alpha)\right] \leq  \frac{1}{f_{\max}}
\leq \mathbb{E}\left[Y_i+z_{\max}(Y_i,\alpha)\right],\!\!\nonumber
\end{align}
and  $\lambda$ is given by
\begin{align}\label{thm6_eq5}
\lambda = \frac{\mathbb{E} [Y_i+z_{\max}(Y_i,\alpha)] -\frac{1}{f_{\max}}  }{\mathbb{E} [z_{\max}(Y_i,\alpha)-z_{\min}(Y_i,\alpha)]}.
\end{align}
\end{theorem}

\begin{proof}
We use \cite[Prop. 6.2.5]{Bertsekas2003} to find a \emph{geometric multiplier} for  \eqref{eq_SD1}. This suggests  that the duality gap between \eqref{eq_SD1} and \eqref{eq_dual} must be zero, because otherwise there exists no geometric multiplier \cite[Prop. 6.2.3(b)]{Bertsekas2003}. The details are provided in Appendix \ref{app_thm_strong_duality}.
\end{proof}
By choosing 
\begin{align}\label{eq_threshold_insight}
\beta = \bar p_{\text{opt}_1}+\alpha,  
\end{align}
Theorem \ref{thm2} follows from Theorem \ref{thm6_strong_duality}.

We note that the extension of Dinkelbach's method in Lemma \ref{lem_ratio_to_minus} and the geometric multiplier technique used in Theorem \ref{thm6_strong_duality} are the key technical tools that make it possible to simplify \eqref{eq_problem2} as the convex optimization problem in \eqref{eq_opt_stopping22}.  These technical tools were also used in a recent study \cite{SunTIT2018}, where a quite different sampling problem is solved.  Further, \eqref{eq_threshold_insight} implies that 
the optimal threshold $\beta$ is equal to the optimum objective value of the MDP $\bar p_{\text{opt}_1}$ plus the optimal Lagrangian dual variable $\alpha$. By using these results, bisection search algorithms are developed in Section \ref{sec_main_results} to compute $\beta$, and the curse of dimensionality  is circumvented. 


\subsection{Proofs of Other Continuous-time Sampling Results}\label{sec_continuous}
Theorem \ref{thm1} follows immediately from Theorem \ref{thm2}, because it is a special case of Theorem \ref{thm2}. 
In particular, because the $Y_i$'s are \emph{i.i.d.}, the optimal objective value to \eqref{eq_problem2} is
\begin{align}\label{eq_KKT_13}
\bar p_{\text{opt}_1}=&\lim_{n\rightarrow \infty}\frac{\sum_{i=1}^{n}\mathbb{E}\left[q(Y_i,z_{\min}(Y_i,0),Y_{i+1})\right]}{\sum_{i=1}^{n} \mathbb{E}\left[Y_i+z_{\min}(Y_i,0)\right]} \nonumber\\
=& \frac{\mathbb{E}\left[q(Y_i,z_{\min}(Y_i,0),Y_{i+1})\right]}{\mathbb{E}\left[Y_i+z_{\min}(Y_i,0)\right]} \nonumber\\
=& \frac{\mathbb{E}\left[\int_{Y_i}^{Y_i+z_{\min}(Y_i,0)+Y_{i+1}} p(t) dt\right]}{\mathbb{E}\left[Y_i+z_{\min}(Y_i,0)\right]},
\end{align}
from which \eqref{thm1_eq22} follows. We note that the root of \eqref{thm1_eq22} must be  unique; otherwise, one can follow the arguments in Appendix \ref{app_thm_strong_duality} to show that the optimal objective value to \eqref{eq_problem2} is non-unique, which cannot be true. Further, as shown in Appendix \ref{app_thm_strong_duality}, the condition ``$\mathbb{E} [p(t+Y_i)]<\infty$ for all  finite $t$'' is not needed in the case of Theorem \ref{thm1}.

\subsection{Proofs of  Discrete-time Sampling Results}\label{sec_discrete_proofs}
The proofs of the discrete-time sampling results are quite similar to their continuous-time counterparts. One difference is that the convex optimization problem \eqref{eq_opt_stopping22} of the continuous-time case becomes 
the following integer optimization problem in the discrete-time case: 
\begin{align}\label{eq_opt_stopping23}
\!\!\min_{z\in\mathbb{N}}\mathbb{E}\left[q(y_i,z,Y_{i+1})\!-\! (\bar p_{\text{opt}_1}\!+\!\alpha) (z\!+\!Y_{i+1}) \right]\!,\!\!
\end{align}
where 
\begin{align}\label{eq_sum1}
q(y_i,z,y')= \sum_{t=y}^{y+z+y'-1}\!\!p(t).
\end{align}
By adopting an idea in  \cite[Problem 5.5.3]{Bertsekas}, we  obtain 
\begin{lemma}\label{lem_optimal_eq_opt_stopping23} 
If $p(\cdot)$ is non-decreasing, then the set of optimal solution to \eqref{eq_opt_stopping23} is $\{z_{\min}(y_i,\alpha),z_{\min}(y_i,\alpha)+1,z_{\min}(y_i,\alpha)+2,\ldots,z_{\max}(y_i,\alpha)\}$, where  
\emph{
\begin{align}\label{eq_gamma2}
\!\!z_{\min}(y,\alpha) &= \!\inf\{ t\in\mathbb{N}: \mathbb{E}[ p(y+t + Y_{i+1}]\! \geq\! \bar p_{\text{opt}_1}\!+\!\alpha \},\!\! \\
\!\!z_{\max} (y,\alpha) 
&=\!\inf\{ t\in\mathbb{N}: \mathbb{E}[ p(y+t + Y_{i+1}]\! >\! \bar p_{\text{opt}_1}\!+\!\alpha \}.\!\! \label{eq_gamma3}
\end{align}}
\end{lemma}
\begin{proof}
See Appendix \ref{app_lem_optimal_eq_opt_stopping23}.
\end{proof}
By replacing Lemma \ref{lem_optimal_eq_opt_stopping22} with Lemma \ref{lem_optimal_eq_opt_stopping23} and following the proof arguments in Section \ref{sec_analysis}.A-\ref{sec_analysis}.F,  the discrete-time optimal sampling results can be proven readily. 

\begin{figure}
	\centering
	\includegraphics[width=0.3\textwidth]{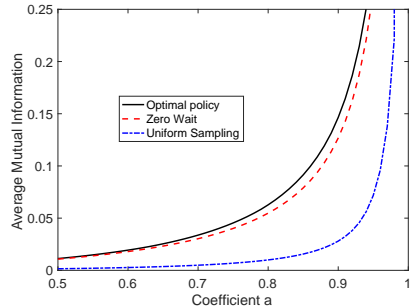}   
	\caption{Average mutual information of the Gauss-Markov source versus the coefficient $a$ in \eqref{eq_markov}, where the service times $Y_i$ are equal to either $1$ or $21$ with probability 0.5.}
	\label{figure4_unconstrained_gaussian_discrete}
\end{figure} 

\begin{figure}
	\centering
	\includegraphics[width=0.3\textwidth]{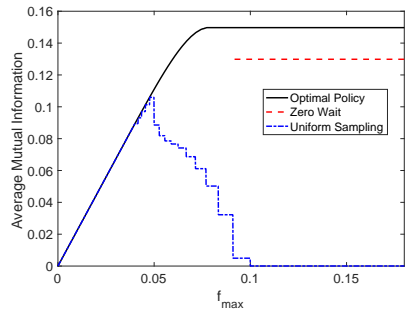}   
	\caption{Average mutual information of the Gauss-Markov source versus $f_{\max}$, where the service times $Y_i$ are equal to either $1$ or $21$ with probability 0.5.}
	\label{figure2_constrained_gaussian_discrete}
\end{figure}

\section{Numerical Results}\label{sec:numerical}


In this section, we compare the age performance of  the following three sampling policies:
\begin{itemize}
\item \textit{Uniform sampling}: Periodic sampling with a period given by $S_{i+1} - S_i = {1}/{f_{\max}}$ for  continuous-time sampling, or $S_{i+1} - S_i = \lceil{1}/{f_{\max}}\rceil$ for discrete-time sampling where $ \lceil x\rceil$ is the smallest integer larger  than or equal to $x$.
\item \textit{Zero-wait}: The sampling policy in \eqref{eq_Zero_wait}, which is infeasible when $f_{\max} <1/\mathbb{E}[Y_i].$ 
\item \textit{Optimal policy}: The sampling policy given by Theorem~\ref{thm2} for  continuous-time sampling, or Theorem~\ref{thm4} for discrete-time sampling. 
\end{itemize}
As the numerical results for continuous-time sampling have been reported in our earlier work \cite{AgeOfInfo2016}, we will focus on the case of discrete-time sampling.


In Fig. \ref{figure4_unconstrained_gaussian_discrete}, we plot  the time-average expected mutual information of the Gauss-Markov source versus the coefficient $a$ in \eqref{eq_markov}, where  $f_{\max}= 0.095$ and $Y_i$ is equal to either $1$ or $21$ with probability 0.5. Hence, $\mathbb{E}[Y_i] = 11$ and the zero-wait sampling policy is infeasible when 
 $f_{\max} <1/11.$
Figure \ref{figure2_constrained_gaussian_discrete} depicts the tradeoff between the time-average expected mutual information of the Gauss-Markov source $X_t$ in \eqref{eq_markov} and $f_{\max}$, where the mutual information is given by \eqref{eq_MI_Gauss_Markov} with $a = 0.9$. 
As the coefficient $a$ grows from 0 to 1, the source $X_t$ becomes more correlated over time. Therefore, the amount of mutual information grows with respect to $a$.  
In addition, the mutual information of the optimal sampling policy is higher than that of zero-waiting sampling and uniform sampling. When $f_{\max}$ is large, the queue length is high and the samples become stale during the long waiting time in the queue. As a result, uniform sampling is far from optimal for large $f_{\max}$. 
\begin{figure}
	\centering
	\includegraphics[width=0.3\textwidth]{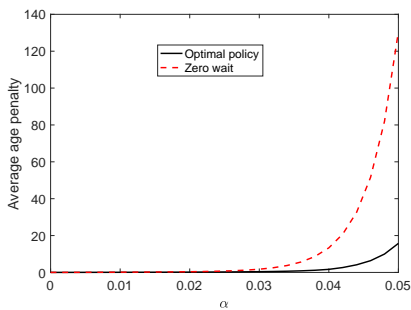}   
	\caption{Average age penalty of an exponential penalty function $p_{\exp}(\Delta_t) = e^{\alpha\Delta_t}-1$ versus the coefficient $\alpha$, where the service times $Y_i$ follow a discretized log-normal distribution.}
	\label{figure10_lognormal_differentpenalties_exp}
\end{figure} 

\begin{figure}
	\centering
	\includegraphics[width=0.3\textwidth]{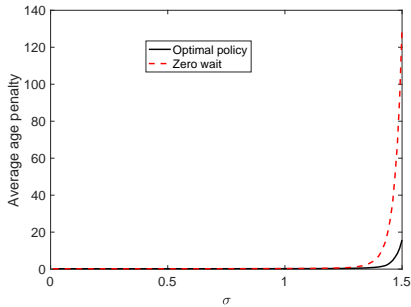}   
	\caption{Average age penalty of an exponential penalty function $p_{\exp}(\Delta_t) = e^{\alpha\Delta_t}-1$ versus the coefficient $\sigma$ of  discretized log-normal service time distribution.}
	\label{figure4_lognormal_differentdistributions}
\end{figure}

Figure \ref{figure10_lognormal_differentpenalties_exp} illustrates the time-average expectation of an exponential penalty function $p_{\exp}(\Delta_t) = e^{\alpha\Delta_t}-1$ versus the coefficient $\alpha$, where $Y_i$ follows a discretized log-normal distribution. In particular, $Y_i$ can be expressed as $Y_i = \lceil e^{\sigma X_i}/ \mathbb{E}[e^{\sigma X_i}]\rceil$, where the $X_i$'s are \emph{i.i.d.} Gaussian random variables with zero mean and unit variance, and $\sigma=1.5$. Figure \ref{figure4_lognormal_differentdistributions} shows the time-average expectation of $p_{\exp}(\Delta_t)$ versus the coefficient $\sigma$ of the discretized log-normal service time distribution. If $\alpha=0$, $p_{\exp}(\Delta_t)$ is a constant function. If $\sigma=0$, the service time $Y_i$ is constant. Corollary \ref{coro3} tells us that zero-waiting sampling is optimal in these two cases, which is in consistent  with Figs. \ref{figure10_lognormal_differentpenalties_exp}-\ref{figure4_lognormal_differentdistributions}. On the other hand, if $\alpha$ and $\sigma$ are large,  one can observe from Figs. \ref{figure10_lognormal_differentpenalties_exp}-\ref{figure4_lognormal_differentdistributions} that zero-waiting sampling is far from optimal. Hence, zero-wait sampling is far from optimal if the age penalty function grows quickly with the age (i.e., $\alpha$ relatively is large) or the service times $Y_i$ are highly random. 
\section{Conclusion}\label{sec_conclusion}
In this paper, we have studied a sampling problem, where  samples are taken from a data source and sent to a remote receiver that is in need of fresh data. We have developed the optimal sampling policies that maximize various data freshness metrics subject to a sampling rate constraint. These sampling policies have nice structures and are easy to compute. Their optimality is established under quite general conditions. Our numerical results show that the optimal sampling policies can be much better than zero-wait sampling and the classic uniform sampling.

\appendices
\section{Proof of Lemma \ref{lem1}}\label{app_lem1}
If the $S_i$'s are independent of $\{X_t, t\geq0\}$, the sampling times $\{S_i: D_i \leq t\}$ of delivered packet contain no information about $X_t$. In addition, because $X_t$ is a Markov chain, $X_{\max\{S_{i} :  D_{i} \leq t\}}=X_{t-\age_t}$ contains all the information in $\bm{W}_t = \{(X_{S_i},S_i): D_i \leq t\}$ about $X_t$. In other words, $X_{t-\age_t}$ is a sufficient statistic of $\bm{W}_t$ for inferring $X_t$. 
Then,  \eqref{eq_lem1} follows from  \cite[Eq. (2.124)]{Cover}.

Next, because $X_t$ is stationary, $I(X_t; X_{t-\age}) = I(X_{\age}; X_{0})$ for all $t$, which is a function of  $\age$. Further, because $X_t$ is a Markov chain, owing to the data processing inequality \cite[Theorem 2.8.1]{Cover}, $I(X_{\age}; X_{0})$ is non-increasing in $\age$. Finally,  mutual information  is non-negative. This completes the proof.

\section{Proof of Lemma \ref{lem_optimal_eq_opt_stopping22}}\label{app_lem_optimal_eq_opt_stopping22}
The one-sided derivative of a function $h$ in the direction of $w$ at $z$ is denoted as 
\begin{align}\label{derivative}
\delta h(z;w)\triangleq& \lim_{\epsilon\rightarrow 0^+} \frac{h(z+\epsilon w)-h(z)}{\epsilon}.
\end{align}
Because the function $h(z)=\mathbb{E}\left[q(y_i,z,Y_{i+1})\right]$ is convex, the one-sided derivative $\delta h(z;w)$  of $h(z)$ exist \cite[p. 709]{Bertsekas}.
Because $z \rightarrow q(y_i,z,y')$ is convex, the function
$\epsilon\rightarrow [q(y_i,z+\epsilon w,y')-q(y_i,z,y')]/\epsilon$ is non-decreasing and bounded from above on $(0,a]$ for some $a>0$ \cite[Proposition 1.1.2(i)]{infinite_dimensional}. 
By monotone convergence theorem \cite[Theorem 1.5.6]{Durrettbook10}, we can interchange the limit and integral operators in $\delta h(z;w)$ such that
\begin{align}
\delta h(z;w)=&\lim_{\epsilon\rightarrow 0^+}\frac{1}{\epsilon}\mathbb{E}\left[ q(y_i,z+\epsilon w,Y_{i+1})- q(y_i,z,Y_{i+1})\right]\nonumber\\
=&\mathbb{E}\!\left[ \lim_{\epsilon\rightarrow 0^+}\frac{1}{\epsilon}\big\{q(y_i,z+\epsilon w,Y_{i+1})- q(y_i,z,Y_{i+1})\big\}\right]\nonumber\\
=& \mathbb{E}\left[\lim\limits_{t\rightarrow z^+}p(y_i+t+Y_{i+1})w 1_{\{w>0\}}\right.\nonumber\\
&~\left.+\lim\limits_{t\rightarrow z^-}p(y_i+t+Y_{i+1})w 1_{\{w<0\}}\right]\nonumber\\
=& \lim\limits_{t\rightarrow z^+}\mathbb{E}\left[p(y_i+t+Y_{i+1})w 1_{\{w>0\}}\right]\nonumber\\
&~+\lim\limits_{t\rightarrow z^-}\mathbb{E}\left[p(y_i+t+Y_{i+1})w 1_{\{w<0\}}\right]\!\!,\!\!\!\label{eq_inequal3}
\end{align}
where $1_{A}$ is the indicator function of event $A$.
According to \cite[p. 710]{Bertsekas} and the convexity of $h(z)$, $z$ is an optimal solution to \eqref{eq_opt_stopping22} if and only if the following assertion is true:
If $z>0$, then 
\begin{align}\label{eq_inequal0}
\delta h(z;w) -(\bar p_{\text{opt}_1}+\!\alpha)w \geq 0,~\forall~w\in \mathbb{R},
\end{align}
otherwise, $z=0$. 
Because $w$ in  \eqref{eq_inequal0} is an arbitrary real number, if we choose $w=1$, then \eqref{eq_inequal0} becomes  
\begin{align}\label{eq_inequal1}
\lim\limits_{t\rightarrow z^+}\mathbb{E}\left[ p(y_i+t+Y_{i+1})\right]\! - \!(\bar p_{\text{opt}_1}+\!\alpha)\geq0.\!\!
\end{align}
Similarly, if we choose $w=-1$, then \eqref{eq_inequal0} implies 
\begin{align}
\lim_{t\rightarrow z^-} \mathbb{E}\left[ p(y_i+t+Y_{i+1})\right]\! -\! (\bar p_{\text{opt}_1}+\!\alpha)\leq0.\!\!\label{eq_inequal2}
\end{align}
Because $p(\cdot)$ is non-decreasing, we can obtain from \eqref{eq_inequal0}-\eqref{eq_inequal2} that if $z>0$, then $z$ satisfies \eqref{eq_inequal3_1}-\eqref{eq_inequal4}: 
\begin{align}\label{eq_inequal3_1}
&\mathbb{E}\left[ p(y_i+t+Y_{i+1})\right]\! - \!(\bar p_{\text{opt}_1}+\!\alpha)\geq0, \text{ if } t>z, \\
&\mathbb{E}\left[ p(y_i+t+Y_{i+1})\right]\! -\! (\bar p_{\text{opt}_1}+\!\alpha)\leq0, \text{ if } t<z, \label{eq_inequal4}
\end{align}
otherwise, $z=0$. The smallest $z$  satisfying \eqref{eq_inequal3_1}-\eqref{eq_inequal4} is 
\begin{align}
z_{\min}(y_i,\alpha) = \inf\{t\geq 0: \mathbb{E}\left[ p(y_i+t+Y_{i+1})\right] \geq \bar p_{\text{opt}_1}+\!\alpha\},\nonumber
\end{align}
and the largest $z$  satisfying \eqref{eq_inequal3_1}-\eqref{eq_inequal4} is 
\begin{align}
z_{\max}(y,\alpha)=&\sup\{t\geq 0: \mathbb{E}\left[ p(y_i+t+Y_{i+1})\right] \leq \bar p_{\text{opt}_1}+\!\alpha\} \nonumber\\
=&\inf\{t\geq 0: \mathbb{E}\left[ p(y_i+t+Y_{i+1})\right] > \bar p_{\text{opt}_1}+\!\alpha\}.\nonumber
\end{align}
Hence, the set of optimal solutions to \eqref{eq_opt_stopping22}  is given by Lemma \ref{lem_optimal_eq_opt_stopping22}. This completes the proof. 
\section{Proof of Theorem \ref{thm6_strong_duality}}\label{app_thm_strong_duality}

According to \cite[Prop. 6.2.5]{Bertsekas2003}, if we can find $\pi^{\star} = (Z_1^{\star},Z_2^{\star},\ldots)$ and $\alpha^{\star}$ that satisfy
 the following conditions:
\begin{align}
&\pi^{\star}\in\Pi_1, \lim_{n\rightarrow \infty} \frac{1}{n} \sum_{i=0}^{n-1} \mathbb{E}\left[Y_i+Z_i^{\star}\right] - \frac{1}{f_{\max}} \geq  0,\label{eq_mix_or_not0}\\
&\alpha^{\star}\geq 0,\label{eq_mix_or_not1}\\
&L(\pi^{\star};\alpha^{\star}) = \inf_{\pi\in\Pi_1}  L(\pi;\alpha^{\star}),\label{eq_mix_or_not}\\
&\alpha^\star \left\{\lim_{n\rightarrow \infty} \frac{1}{n} \sum_{i=0}^{n-1} \mathbb{E}\left[Y_i+Z_i^{\star}\right] - \frac{1}{f_{\max}}\right\} = 0,\label{eq_KKT_last}
\end{align}
then $\pi^{\star}$ is an optimal solution to \eqref{eq_SD1} and $\alpha^\star$ is a geometric multiplier \cite{Bertsekas2003} for  \eqref{eq_SD1}. 
Further, if we can find such $\pi^{\star}$ and $\alpha^{\star}$, then the duality gap between \eqref{eq_SD1} and \eqref{eq_dual} must be zero, because otherwise there is no geometric multiplier \cite[Prop. 6.2.3(b)]{Bertsekas2003}. 
The remaining task is to find $\pi^{\star}$ and $\alpha^{\star}$ that satisfy \eqref{eq_mix_or_not0}-\eqref{eq_KKT_last}.

According to Lemma \ref{lem_gamma}, the set of optimal solutions to \eqref{eq_mix_or_not} is given by $\Gamma (\alpha^\star)$. Hence, we only need to find $\alpha^\star$ and $\pi^\star\in\Gamma (\alpha^\star)$  that 
satisfy \eqref{eq_mix_or_not0}, \eqref{eq_mix_or_not1}, and \eqref{eq_KKT_last}. The search for such  $\alpha^\star$ and $\pi^\star$ falls into the following two cases:


\textbf{Case 1:} If \eqref{thm6_eq0} is satisfied, then $\alpha^\star_1  = 0$ and $\pi^\star_1 = (z_{\min}(Y_1,0), z_{\min}(Y_2,0), \ldots)$ satisfy the conditions \eqref{eq_mix_or_not0}-\eqref{eq_KKT_last}. 

\textbf{Case 2:} If \eqref{thm6_eq0} is not satisfied, we seek $\alpha^\star_2\geq0$ and $\pi^\star_2= (Z_1^{\star},Z_2^{\star},\ldots)\in\Gamma (\alpha^\star_2)$ that satisfy
\begin{align}\label{eq_KKT_2}
\lim_{n\rightarrow \infty} \frac{1}{n} 
\sum_{i=0}^{n-1} \mathbb{E}\left[Y_i+Z_i^{\star}\right] = \frac{1}{f_{\max}}.
\end{align}
By  Lemma \ref{lem_gamma}, we can get  from \eqref{eq_KKT_2} that
\begin{align}\label{eq_KKT_10}
\lim_{n\rightarrow \infty} \frac{1}{n} 
\sum_{i=0}^{n-1} \mathbb{E}\left[Y_i+z_{\min}(Y_i,\alpha^\star_2)\right] \leq  \frac{1}{f_{\max}}~~~\nonumber\\
\leq\lim_{n\rightarrow \infty} \frac{1}{n} 
\sum_{i=0}^{n-1} \mathbb{E}\left[Y_i+z_{\max}(Y_i,\alpha^\star_2)\right].
\end{align}
Because the $Y_i$'s are \emph{i.i.d.}, \eqref{eq_KKT_10} is equivalent to 
\begin{align}\label{eq_KKT_11}
\!\!\mathbb{E}\left[Y_i+z_{\min}(Y_i,\alpha^\star_2)\right] \leq  \frac{1}{f_{\max}}
\leq \mathbb{E}\left[Y_i+z_{\max}(Y_i,\alpha^\star_2)\right].\!\!
\end{align}

Next, we will find $\alpha^\star_2\geq 0$ that satisfies \eqref{eq_KKT_11}. According to \eqref{eq_gamma}-\eqref{eq_gamma1}, $z_{\min}(y,\alpha)$ and $z_{\max}(y,\alpha)$ are non-decreasing in $\alpha$. Hence, $\mathbb{E} [z_{\min}(Y_i,\alpha)]$ and $\mathbb{E} [z_{\max}(Y_i,\alpha)]$ are also non-decreasing in $\alpha$. 
In addition, it holds that for all $\alpha_0> 0$
\begin{align}
&\lim_{\alpha \rightarrow \alpha_0^-} z_{\max}(y,\alpha) = z_{\min}(y,\alpha_0)\!\nonumber\\
&~~~~~~~~~~~~~~ \leq z_{\max}(y,\alpha_0) =\lim_{\alpha \rightarrow \alpha_0^+} z_{\min}(y,\alpha).
\end{align}
By invoking the monotone convergence theorem \cite[Theorem 1.5.6]{Durrettbook10}, we obtain that for all $\alpha_0> 0$
\begin{align}\label{eq_KKT_16}
&\lim_{\alpha \rightarrow \alpha_0^-} \mathbb{E} [z_{\max}(Y_i,\alpha)] = \mathbb{E} [z_{\min}(Y_i,\alpha_0)]\nonumber\\
&~~~~~~~~~~~~~\leq \mathbb{E} [z_{\max}(Y_i,\alpha_0)]=\lim_{\alpha \rightarrow \alpha_0^+} \mathbb{E} [z_{\min}(Y_i,\alpha)] . 
\end{align}
Because $\mathbb{E} [p(t+Y_i)]<\infty$ for all  finite $t$, it holds  for all $y\geq0$ that $z_{\max}(y,\alpha)$ will increase to $\infty$ as $\alpha$ grows from 0 to $\infty$. By invoking the monotone convergence theorem again, we obtain that $\mathbb{E} [z_{\max}(Y_i,\alpha)]$ will increase to $\infty$ as $\alpha$ grows  from 0 to $\infty$. Hence,
\begin{align}\label{eq_KKT_12}
[\mathbb{E} [z_{\min}(Y_i,0)],\infty) = \bigcup\limits_{\alpha\geq0} \Big[\mathbb{E} [z_{\min}(Y_i,\alpha)],\mathbb{E} [z_{\max}(Y_i,\alpha)]\Big]. 
\end{align}
In Case 2, \eqref{thm6_eq0} is not satisfied, which implies
\begin{align}\label{eq_KKT_14}
 \frac{1}{f_{\max}} \in [\mathbb{E} [z_{\min}(Y_i,0)],\infty).
\end{align}
Hence, \eqref{eq_KKT_16}-\eqref{eq_KKT_14} tell us that there exists a unique $\alpha^\star_2 \geq 0$ satisfying \eqref{eq_KKT_11}.   
Further, policy $\pi^\star\in\Gamma (\alpha^\star_2)$ is chosen as
\begin{align}\label{eq_KKT_15}
Z_i^{\star} = \left\{\begin{array}{l l} z_{\min}(Y_i,\alpha^\star_2)\text{ with probability }\lambda,\\ z_{\max}(Y_i,\alpha^\star_2) \text{ with probability }1-\lambda, \end{array}\right.
\end{align}
where $\lambda$ is given by
\begin{align}
\lambda = \frac{\mathbb{E} [Y_i+z_{\max}(Y_i,\alpha^\star_2)] -\frac{1}{f_{\max}}  }{\mathbb{E} [z_{\max}(Y_i,\alpha^\star_2)-z_{\min}(Y_i,\alpha^\star_2)]}.
\end{align}
By combining \eqref{eq_KKT_11}, \eqref{eq_KKT_14}, and \eqref{eq_KKT_15}, \eqref{eq_KKT_2} follows. Hence, the   $\alpha^\star_2$ and $\pi^\star_2$ selected above satisfy the conditions \eqref{eq_mix_or_not0}-\eqref{eq_KKT_last}.

In both cases, \eqref{eq_mix_or_not0}-\eqref{eq_KKT_last} are satisfied. By \cite[Prop. 6.2.3(b)]{Bertsekas2003},  the duality gap between \eqref{eq_SD1} and \eqref{eq_dual} is zero. A solution to \eqref{eq_SD1} and \eqref{eq_dual} is provided in the arguments above. This completes the proof. 

\section{Proof of Corollary \ref{coro}}\label{app_coro}

We note that the zero-wait sampling policy can be expressed as \eqref{thm1_eq12} with $\mathbb{E}\left[p(\text{ess}\inf Y_i +Y_{i+1})\right] \geq \beta $. 

In the one direction, if the zero-wait sampling
policy is optimal, then the root of \eqref{thm1_eq22} must satisfy $\mathbb{E}\left[p(\text{ess}\inf Y_i +Y_{i+1})\right] \geq \beta $. Substituting this into \eqref{thm1_eq12}, yields $D_{i+1}(\beta)= D_i(\beta)+Y_{i+1} = S_i(\beta) + Y_i + Y_{i+1}$. Combining this with \eqref{thm1_eq22}, we get
\begin{align}
\mathbb{E}\left[p(\text{ess}\inf Y_i +Y_{i+1})\right]\geq \beta = \frac{\mathbb{E}\left[\int_{Y_i}^{Y_i+Y_{i+1}}\! p(t)dt\right]}{\mathbb{E}[Y_{i+1}]},
\end{align}
which  implies \eqref{eq_coro1}.

In the other direction, if \eqref{eq_coro1} holds, then by choosing
\begin{align}\label{eq_coro1_2}
\beta = \frac{\mathbb{E}\left[\int_{Y_i}^{Y_i+Y_{i+1}}\! p(t)dt\right]}{\mathbb{E}[Y_{i+1}]},
\end{align}
we get $\mathbb{E}\left[p(\text{ess}\inf Y_i +Y_{i+1})\right] \geq \beta $. According to \eqref{eq_coro1_2}, such a $\beta$ 
is a root of \eqref{thm1_eq22}. Therefore, the zero-wait sampling
policy is optimal. This completes the proof.

\section{Proof of Corollary \ref{coro4}}\label{app_coro4}
We first prove Part (a). 
If $Y_i=y$ almost surely, then 
\begin{align}
\mathbb{E}\left[p(\text{ess}\inf Y_i +Y_{i+1})\right] = p(2y)\geq \frac{\int_{y}^{2y}\! p(t)dt}{y}
\end{align}
holds for all non-decreasing $p(\cdot)$.
Hence, \eqref{eq_coro1} is satisfied and the zero-wait sampling policy is optimal. 

Next, we consider Part (b).
If {$\text{ess}\inf Y_i = 0$}, then
\begin{align}\label{eq_app_coro4}
&\mathbb{E}\left[p(\text{ess}\inf Y_i +Y_{i+1})\right] = \mathbb{E}\left[p(Y_{i+1})\right] = \mathbb{E}\left[p(Y_{i})\right].
\end{align}
Because $\mathbb{E}[Y_{i+1}] = \mathbb{E}[Y_{i}]>0$, then  the event  $Y_{i+1}>0$ has a non-zero probability. Further, because $p(\cdot)$ is strictly increasing, the event 
$p(t) > p(Y_i)$ for $t\in(Y_i, Y_i+Y_{i+1})$ has a non-zero probability. Hence, 
\begin{align}\label{eq_app_coro4_1}
\mathbb{E}\left[\int_{Y_i}^{Y_i+Y_{i+1}}\! p(t)dt\right]>& \mathbb{E}\left[\int_{Y_i}^{Y_i+Y_{i+1}}\! p(Y_i)dt\right] \nonumber\\
=& \mathbb{E}[Y_{i+1}] \mathbb{E}\left[p(Y_{i})\right].
\end{align}
By combining \eqref{eq_app_coro4} and \eqref{eq_app_coro4_1}, \eqref{eq_coro1} is not true and the zero-wait sampling policy is not optimal.
This completes the proof.

\section{Proof of Lemma \ref{lem_optimal_eq_opt_stopping23}}\label{app_lem_optimal_eq_opt_stopping23}

Using \eqref{eq_sum1}, \eqref{eq_opt_stopping23} can be expressed as
\begin{align}\label{eq_opt_stopping4}
\min_{z\in \mathbb{N}} \mathbb{E}\left[\sum_{t=0}^{z+Y_{i+1}-1}\!\![p(t+y_i)- (\bar p_{\text{opt}_1}\!+\!\alpha)] \right].
\end{align}
It holds that  for $m =1, 2,3,\ldots$
\begin{align}\label{eq_opt_stopping5}
&\mathbb{E}\!\!\left[\sum_{t=0}^{m+Y_{i+1}}\!\!\!\![p(t+y_i)- (\bar p_{\text{opt}_1}\!+\!\alpha)] \right. \nonumber\\
&~~~~~~~~~~~~~\left.-\!\!\sum_{t=0}^{m+Y_{i+1}-1}\!\!\!\![p(t+y_i)- (\bar p_{\text{opt}_1}\!+\!\alpha)] \right]\nonumber\\
=& \mathbb{E}\!\left[p(y_i+m+Y_{i+1})- (\bar p_{\text{opt}_1}\!+\!\alpha)\right].
\end{align}
Because $p(\cdot)$ is non-decreasing, if $z$ is chosen according to Lemma \ref{lem_optimal_eq_opt_stopping23}, we can obtain
\begin{align}
&\!\!\mathbb{E}\!\left[p(y_i+t+Y_{i+1})-(\bar p_{\text{opt}_1}\!+\!\alpha) \right] \leq 0,~ t= 0, \ldots,  z-1, \\
&\!\!\mathbb{E}\!\left[p(y_i+t+Y_{i+1})-(\bar p_{\text{opt}_1}\!+\!\alpha)\right] \geq 0,~ t = z, z+1,\ldots\label{eq_opt_stopping6}
\end{align}
Using \eqref{eq_opt_stopping5}-\eqref{eq_opt_stopping6}, one can see that $\{z_{\min}(y_i,\alpha),z_{\min}(y_i,\alpha)+1, $ $z_{\min}(y_i,\alpha)+2,\ldots,z_{\max}(y_i,\alpha)\}$
is the set of optimal solutions to \eqref{eq_opt_stopping23}. 
This completes the proof.

\bibliographystyle{IEEEtran}
\bibliography{ref}

\end{document}